%% file: main.tex
\documentclass[11pt,fancychapters]{report}
\usepackage[a4paper, total={6in, 8in}]{geometry}

\makeatletter
\renewcommand{\maketitle}{
\begin{center}

\pagestyle{empty}
\phantom{.}  
\vspace{3cm}

{\Huge \bf Data-Injection Attacks \par}
\vspace{1.5cm}

{\large I\~naki Esnaola$^{1,3}$, Samir M. Perlaza$^{2,3}$, and Ke Sun$^{1}$}\\[0.2cm]
\vspace{0.1cm}
1. Dept. of Automatic Control and Systems Eng., University of Sheffield,  UK
\vspace{0.1cm}

2. INRIA, Centre de Recherche de Sophia Antipolis-M\'{e}diterran\'{e}e, France
\vspace{0.1cm}

3. Department of Electrical Engineering, Princeton University,  USA

\vspace{0.2cm}
email: esnaola@sheffield.ac.uk, samir.perlaza@inria.fr, ke.sun@sheffield.ac.uk

\vspace{1.5cm}
{\large\textbf{Chapter 9} of:  Advanced Data Analytics for Power Systems, A. Tajer, S. M. Perlaza and H. V. Poor, Eds., Cambridge University Press, Cambridge, UK, 2021,  pp. 197-229}
\vspace{5cm}

{February 02, 2021}

\end{center}
}\makeatother

 \usepackage{amsmath}
  \usepackage{amsthm}
  \usepackage{graphicx}

\usepackage{yhmath}
\usepackage{graphics}
\usepackage{textcomp}
\usepackage{amssymb}
\usepackage{amscd}
\usepackage{dsfont}
\usepackage{epsfig}
\usepackage{amsfonts}
\usepackage{multicol} 
\usepackage{color}
\usepackage{algorithm}
\usepackage[noend]{algorithmic}
\usepackage{amsthm}
\usepackage{bbm}
\usepackage{balance}

\input{macros}

\begin{document}
\maketitle
\pagenumbering{gobble}
\newpage
\pagenumbering{roman}
\tableofcontents
\newpage
\pagenumbering{arabic}









 
 \include{Bayesian_Attacks_IE}
%



 \bibliographystyle{IEEEbib} 
 
 \bibliography{reference_mmse, reference_stealth, reference_rmt}





\end{document}

%% file: macros.tex
\long\def\comment#1{}

\newfont{\bbb}{msbm10 scaled 700}

\newfont{\bb}{msbm10 scaled 1100}

\newcommand{\RR}{\mbox{\bb R}}
\newcommand{\PP}{\mbox{\bb P}}

\newcommand{\ZZ}{\mbox{\bb Z}}

\newcommand{\EE}{\mbox{\bb E}}

\newcommand{\av}{{\bf a}}

\newcommand{\cv}{{\bf c}}

\newcommand{\ev}{{\bf e}}

\newcommand{\uv}{{\bf u}}

\newcommand{\xv}{{\bf x}}
\newcommand{\yv}{{\bf y}}
\newcommand{\zv}{{\bf z}}

\newcommand{\zerov}{{\bf 0}}

\newcommand{\Am}{{\bf A}}
\newcommand{\Bm}{{\bf B}}

\newcommand{\Gm}{{\bf G}}
\newcommand{\Hm}{{\bf H}}
\newcommand{\Id}{{\bf I}}

\newcommand{\Lm}{{\bf L}}
\newcommand{\Mm}{{\bf M}}

\newcommand{\Sm}{{\bf S}}

\newcommand{\Um}{{\bf U}}
\newcommand{\Wm}{{\bf W}}

\newcommand{\Zm}{{\bf Z}}

\newcommand{\Tt}{\text{T}}

\newcommand{\Ac}{{\cal A}}

\newcommand{\Cc}{{\cal C}}

\newcommand{\Hc}{{\cal H}}

\newcommand{\Kc}{{\cal K}}

\newcommand{\Mc}{{\cal M}}
\newcommand{\Nc}{{\cal N}}

\newcommand{\Sc}{{\cal S}}

\newcommand{\RNum}[1]{\uppercase\expandafter{\romannumeral #1\relax}}

\newcommand{\muv}{\hbox{\boldmath$\mu$}}

\newcommand{\Lambdam}{\hbox{\boldmath$\Lambda$}}
\newcommand{\Deltam}{\hbox{\boldmath$\Delta$}}
\newcommand{\Sigmam}{\hbox{\boldmath$\Sigma$}}

\newcommand{\diag}{{\hbox{diag}}}

\newcommand{\trace}{{\hbox{tr}}}

\renewcommand{\arg}{{\hbox{arg}}}
\newcommand{\var}{{\hbox{var}}}

\newcommand{\eqdef}{\stackrel{\Delta}{=}}

\newcommand{\GameNF}{\mathcal{G} = \left(\mathcal{K}, \left\lbrace\mathcal{A}_k \right\rbrace_{k \in \mathcal{K}},\phi \right)}
\newcommand{\gameNF}{\mathcal{G}}
\newcommand{\BR}{\mathrm{BR}}

\def\LRT#1#2{\!
\raisebox{.2ex}{$
{{\scriptstyle\;#1}\atop{\displaystyle\gtrless}}
\atop
{\raisebox{-1.25ex}{$\scriptstyle\;#2$}}
$}
\!}

\newcommand{\be}{\begin{equation}}
\newcommand{\ee}{\end{equation}}

\newtheorem{definition}{Definition}[chapter]

\newtheorem{theorem}{Theorem}[chapter]
\newtheorem{corollary}{Corollary}[chapter]
\newtheorem{lemma}{Lemma}[chapter]
\newtheorem{proposition}{Proposition}[chapter]

%% file: Bayesian_Attacks_IE.tex
\setcounter{chapter}{8}
\chapter{Data-Injection Attacks}\label{Chapter1}



\section{Introduction}
\label{SecIntro}

The pervasive deployment of sensing, monitoring, and data acquisition techniques in modern power systems enables the definition of functionalities and services that leverage accurate and real-time information about the system. This wealth of data supports network operators in the design of advanced control and management techniques that will inevitably change the operation of future power systems. An interesting side-effect of the data collection exercise that is starting to take place in power systems is that the unprecedented data analysis effort is shedding some light on the turbulent dynamics of power systems. While the underlying physical laws governing power systems are well understood, the large scale, distributed structure, and stochastic nature of the generation and consumption processes in the system results in a complex system. The large volumes of data about the state of the system are opening the door to modelling aspirations that were not feasible prior to the arrival of the smart grid paradigm.

The refinement of the models describing the power system operation will undoubtedly provide valuable insight to the network operator. 
However, that knowledge and the explanatory principles that it uncovers are also subject to be used in a malicious fashion. 
Access to statistics describing the state of the grid can inform malicious attackers by allowing them to pose the data-injection problem \cite{liu_false_2009} problem within a probabilistic framework \cite{kosut_malicious_2011, esnaola_maximum_2016}. 
By describing the processes taking place in the grid as a stochastic process, the network operator can incorporate the statistical description of the state variables in the state estimation procedure and pose it within a Bayesian estimation setting. 
Similarly, the attacker can exploit the stochastic description of the state variables by incorporating it to the attack construction in the form of prior knowledge about the state variables. 
Interestingly wether the network operator or the attacker benefit more from adding a stochastic description to the state variables does not have a simple answer and depends greatly on the parameters describing the power system.

In this chapter we review some of the basic attack constructions that exploit a stochastic description of the state variables. We pose the state estimation problem in a Bayesian setting and cast the bad data detection procedure as a Bayesian hypothesis testing problem. This revised detection framework provides the benchmark for the attack detection problem that limits the achievable attack disruption. Indeed, the trade-off between the impact of the attack, in terms of disruption to the state estimator, and the probability of attack detection is analytically characterized within this Bayesian attack setting. 
We then generalize the attack construction by considering information-theoretic measures that place fundamental limits to a broad class of detection, estimation, and learning techniques. Because the attack constructions proposed in this chapter rely on the attacker having access to the statistical structure of the random process describing the state variables, we conclude by studying the impact of imperfect statistics on the attack performance. Specifically, we study the attack performance as a function of the size of the training data set that is available to the attacker to estimate the second-order statistics of the state variables.

\section{System Model}

\subsection{Bayesian State Estimation}
We model the state of the system as the vector of $n$ random variables $X^n$ taking values in $\mathbb{R}^{n}$ with distribution $P_{X^n}$. The random variable $X_i$ with $i=1, 2, \ldots, n$, denotes the state variable $i$ of the power system, and therefore, each entry represents a different physical magnitude of the system that the network operator wishes to monitor. 
The prior knowledge that is available to the network operator is described by the probability distribution $P_{X^n}$. The knowledge of the distribution is a consequence of the modelling based on historical data acquired by the network operator.  
Assuming linearized system dynamics with $m$ measurements corrupted by additive white Gaussian noise (AWGN), the measurements are modelled as the vector of random variables $Y^m\in\mathbb{R}^m$ with distribution $P_{Y^m}$ given by
\begin{align}
Y^n=\Hm X^n+Z^m,
\end{align}
where $\Hm\in\mathbb{R}^{m \times n}$ is the Jacobian of the linearized system dynamics around a given operating point and $Z^m\thicksim\mathcal{N}(0,\sigma^2\mathbf{I})$ is thermal white noise with power spectral density $\sigma^2$. While the operation point of the system induces a dynamic on the Jacobian matrix $\Hm$, in the following we assume that the time-scale over which the operation point changes is small compared to the time-scale at which the state estimator operates to produce the estimates. For that reason, in the following we assume that the Jacobian matrix is fixed and the only sources of uncertainty in the observation process originate from the stochasticity of the state variables and the additive noise corrupting the measurements. 


The aim of the state estimator is to obtain an estimate $\hat{X}^n$ of the state vector $X^n$ from the system observations $Y^m$. In this chapter we adopt a linear estimation framework resulting in an estimate given by $\hat{X}^n = \Lm Y^m$, where $\Lm\in\mathbb{R}^{n\times m}$ is the linear estimation matrix determining the estimation procedure.
In the case in which the operator knows the distribution $P_{X^n}$ of the underlying random process governing the state of the network, the estimation is performed by selecting the estimate that minimizes a given error cost function. A common approach is to use the mean square error (MSE) as the error cost function. In this case, the network operator uses an estimator $\Mm$ that is the unique solution to the following optimization problem:
\be
\Mm = \arg\min_{\Lm \in \mathbb{R}^{n \times m}} {\mathbb{E}}\left [\frac{1}{n}\|X^n - \Lm Y^m \|_2^2\right],
\ee
where the expectation is taken with respect to $X^n$ and $Z^m$.
%

Under the assumption that the network state vector $X^n$ follows an $n$-dimensional real Gaussian distribution with zero mean and covariance matrix $\Sigmam_{X\!X}\in \Sc^{m}_{+}$, i.e. $X^n\thicksim\Nc(\mathbf{0},\Sigmam_{X\!X})$, the minimum MSE (MMSE) estimate is given by 
\begin{align}
\label{EqMMSEe}
\hat{X}^n & \eqdef \mathbb{E}[X^n|Y^m]=\Mm Y^m
\end{align}
where,
\be
\label{EqMMSEm}
\Mm=\Sigmam_{X\!X} \Hm^{\sf{T}}(\Hm\Sigmam_{X\!X}\Hm^{\sf{T}}+\sigma^2\mathbf{I})^{-1}.
\ee

\subsection{Deterministic Attack Model}

The aim of the attacker is to corrupt the estimate by altering the measurements. Data-injection attacks alter the measurements available to the operator by adding an attack vector to the measurements.
The resulting observation model with the additive attack vector is given by
\begin{align}
\label{eq:obs_mod_det_a}
Y^{m}_{a} = \Hm X^{m} + Z^{m} + \av,
\end{align}
where $\av^{m} \in \mathbb{R}^{m}$ is the attack vector and $Y^{m}_{a} \in \mathbb{R}^{m } $ is the vector containing the compromised measurements \cite{liu_false_2009}. Note that in this formulation, the attack vector does not have a probabilistic structure, i.e. the attack vector is deterministic. The random attack construction is considered later in the chapter. 

The intention of the attacker can respond to diverse motivations, and therefore, attack construction strategy changes depending on the aim of the attacker. In this chapter, we study attacks that aim to maximize the monitoring disruption, i.e. attacks that obstruct the state estimation procedure with the aim of deviating the estimate as much as possible from the true state. In that sense, the attack problem is bound to the cost function used by the state estimator to obtain the estimate, as the attacker aims to maximize it while the estimator aims to minimize it.  In the MMSE setting described in the preceding text, it follows that the the impact of the attack vector is obtained by noticing that the estimate when the attack vector is present is given by
\begin{align}
\label{EqMMSEe1}
\hat{X}^n_{a} & = \Mm (\Hm X^n + Z^m) + \Mm \av.
\end{align}
The term $\Mm \av$ is referred to as the \emph{Bayesian injection vector} introduced by the attack vector $\av$ and is denoted by
\be
\label{eq:xv_a}
\cv \eqdef   \Mm \av = \Sigmam_{X\!X}\Hm^{\sf{T}}(\Hm\Sigmam_{X\!X}\Hm^{\sf{T}}+\sigma^2\mathbf{I})^{-1}\av.
\ee
The {\it Bayesian injection vector} is a deterministic vector that corrupts the MMSE estimate of the operator resulting in 
\be
\hat{X}^n_{a} =\hat{X}^n+\cv.
\ee
where $\hat{X}^n$ is given in (\ref{EqMMSEe}).

\subsection{Attack Detection}
As a part of the grid management, a network operator systematically attempts to identify measurements that are not deemed of sufficient quality for the state estimator. In practice, this operation can be cast as a hypothesis testing problem with hypotheses
\begin{align}
\mathcal{H}_0:&\quad\textnormal{There is no attack}, \quad \textnormal{versus} \nonumber \\
\mathcal{H}_1:&\quad\textnormal{Measurements are compromised}.
\end{align}
Assuming the operator knows the distribution of the state variables, $P_{X^n}$, and the observation model (\ref{eq:obs_mod_det_a}), then it can obtain the joint distribution of the measurements and the state variables for both normal operation conditions and the case when an attack is present, i.e. $P_{X^nY^m}$ and $P_{X^nY^m_a}$, respectively. 

Under the assumption that the state variables follow a multivariate Gaussian distribution  $X^n\thicksim\mathcal{N}(\mathbf{0},\Sigmam_{X\!X})$ it follows that the vector of measurements $Y^n$ follows an $m$-dimensional real Gaussian random distribution with covariance matrix
\be
\Sigmam_{Y\!Y} = \Hm\Sigmam_{X\!X}\Hm^{\sf{T}}+\sigma^2\mathbf{I},
\ee
and mean $\av$ when there is an attack; or zero mean when there is no attack. Within this setting, the hypothesis testing problem described before is adapted to the attack detection problem by comparing the following hypotheses:
\begin{align}
\mathcal{H}_0:&\quad Y^m\thicksim\mathcal{N}(\mathbf{0},\Sigmam_{Y\!Y}), \quad \textnormal{versus} \nonumber\\
\mathcal{H}_1:&\quad Y^m\thicksim\mathcal{N}(\av,\Sigmam_{Y\!Y}).
\end{align}
A worst case scenario approach is assumed for the attackers, namely, the operator knows the attack vector, $\av$, used in the attack. However, the operator does not know a priori whether  the grid is under attack or not, which accounts for the need of an attack detection strategy. That being the case, the optimal detection strategy for the operator is to perform a likelihood ratio test (LRT) $L(\yv, \av)$ with respect to the observations $\yv$. Under the assumption that state variables follow a multivariate Gaussian distribution, the likelihood ratio can be calculated as
\be
\label{eq:MAP_detect}
L(\yv, \av)= \frac{f_{\mathcal{N}(\mathbf{0},\mathbf{\Sigma}_{\yv\yv})}(\yv)}{f_{\mathcal{N}(\av,\mathbf{\Sigma}_{\yv\yv})}(\yv )}=  \exp\left( \frac{1}{2} \av^{\sf{T}} \Sigmam_{Y\!Y}^{-1} \av  -  \av^{\sf{T}} \Sigmam_{Y\!Y}^{-1}  \yv  \right),
\ee
where $f_{\mathcal{N}\left(\mathbf{\muv},\mathbf{\Sigma}\right)}$ is the probability density function of a multivariate Gaussian random vector with mean $\muv$ and covariance matrix $\Sigmam$.
Therefore, either hypothesis is accepted by evaluating the inequalities
\be
\label{eq:MAP_decision}
L(\yv, \av)  \mathop{\gtrless}_{\mathcal{H}_1}^{\mathcal{H}_0}\tau,
\ee
where $\tau\in[0,\infty)$ is tuned to set the trade-off between the probability of detection and the probability of false alarm.
%

\section{Centralized Deterministic Attacks}
\label{sec:C_attacks}
This section describes the construction of data-injection attacks in the case in which there is a unique attacker with access to all the measurements on the power system. This scenario is referred to as \emph{centralized attacks} in order to highlight that there exists a unique entity deciding the data-injection vector $\av \in \mathbb{R}^{m}$ in \eqref{eq:obs_mod_det_a}. The difference between the scenario in which there exists a unique attacker or several (competing or cooperating) attackers is subtle and it is treated in Section \ref{SecDA}.

Let $\Mc = \lbrace 1, \ldots, m\rbrace$ denote the set of all $m$ sensors available to the network operator. A sensor is said to be compromised if the attacker is able to arbitrarily modify its output.
Given a total energy budget $E > 0$ at the attacker, the set of all possible attacks that can be injected to the network can be explicitly described:
\be
\Ac = \left\lbrace \av  \in \mathbb{R}^{m}: \, \av^{\sf{T}} \av \leqslant E     \right\rbrace.
\ee

\subsection{Attacks with Minimum Probability of Detection}

The attacker chooses a vector $\av \in \Ac$ taking into account the trade-off between the probability of being detected and the distortion induced by the Bayesian injection vector given by \eqref{eq:xv_a}.
However, the choice of a particular data-injection vector is
not trivial as the attacker does not have any information about the exact realizations of the vector of state variables $\xv$ and the noise vector $\zv$. A reasonable assumption on the knowledge of the attacker is to consider that it knows the structure of the power system and thus, it knows the matrix $\Hm$. It is also reasonable to assume that it knows the first and second moments of the state variables $X^{n}$ and noise $Z^{m}$ as this can be computed from historical data.

Under these knowledge assumptions, the probability that the network operator is unable to detect the attack vector $\av$ is
\be
\label{EqProbND0}
\sf{P_{ND}}(\av) = \sf{E} \left[ \mathds{1}_{\left\lbrace L(\yv, \av) > \tau \right\rbrace} \right],
\ee
where the expectation is taken over the joint probability distribution of state variables $X^n$ and the AWGN noise vector $Z^n$, and $\mathds{1}_{\lbrace\cdot\rbrace}$ denotes the indicator function. Note that under these assumptions, $Y^m$ is a random variable with Gaussian distribution with mean $\av$ and covariance matrix $\Sigmam_{Y\!Y}$.
Thus, the probability $\sf{P_{ND}}(\av)$ of a vector $\av$ being a successful attack, i.e., a non-detected attack is given by \cite{Poor_springer_88}
\be
\label{EqProbND}
\sf{P_{ND}}(\av) = \frac{1}{2} \sf{erfc} \left(\frac{\frac{1}{2} \av^{\sf{T}}\Sigmam_{Y\!Y}^{-1} \av + \log\tau }{\sqrt{2 \av^{\sf{T}}\Sigmam_{Y\!Y}^{-1} \av }} \right).
\ee

Often, the knowledge of the threshold $\tau$ in \eqref{eq:MAP_decision} is not available to the attacker and thus, it cannot determine the exact probability of not being detected for a given attack vector $\av$. However, the knowledge of whether $\tau > 1$ or $\tau \leqslant 1$ induces different behaviors on the attacker. The following propositions follow immediately from \eqref{EqProbND} and the properties of the complementary error function.
\begin{proposition}[Case $\tau \leqslant 1$]\label{CorollaryTauSmall}
Let $\tau \leqslant 1$. Then, for all $\av \in \Ac$, $\sf{P_{ND}}(\av) < \sf{P_{ND}}\left( (0, \ldots, 0) \right)$ and  the probability $\sf{P_{ND}}(\av)$  is monotonically decreasing with $\av^{\sf{T}}\Sigmam_{Y\!Y}^{-1} \av$.
\end{proposition}

\begin{proposition}[Case $\tau > 1$]\label{CorollaryTauBig}
Let $\tau > 1$ and let also $\Sigmam_{Y\!Y} = \Um_{Y\!Y} \Lambdam_{Y\!Y} \Um_{Y\!Y}^{\sf{T}}$ be the singular value decomposition of $\Sigmam_{Y\!Y}$, with $\Um_{Y\!Y}^{\sf{T}} = \left( \uv_{Y\!Y,1}, \ldots, \uv_{Y\!Y,m}\right)$ and $\Lambdam_{Y\!Y} = \diag\left( \lambda_{Y\!Y,1}, \ldots, \lambda_{Y\!Y,m}\right)$ and  $\lambda_{Y\!Y,1} \geqslant\lambda_{Y\!Y,2} \geqslant  \ldots, \geqslant \lambda_{Y\!Y,m}$. Then, any vector of the form
\be
\label{EqLowestProbAttack}
\av = \pm \sqrt{ \lambda_{Y\!Y,k}  2\log\tau} \uv_{Y\!Y,k},
\ee
with $k \in \lbrace 1, \ldots, m \rbrace$, is a data-injection attack that satisfies for all $\av' \in \mathbb{R}^{m}$, $\sf{P_{ND}}(\av') \leqslant \sf{P_{ND}}(\av)$.
\end{proposition}

The proof of Proposition~\ref{CorollaryTauSmall} and Proposition~\ref{CorollaryTauBig} follows.

\begin{proof}
Let $x = \av^{\sf{T}}\Sigmam_{Y\!Y}^{-1} \av$ and note that $x > 0 $ due to the positive definiteness of $\Sigmam_{Y\!Y}$. Let also the function $g: \mathbb{R} \rightarrow \mathbb{R}$ be
\be
g(x) = \frac{\frac{1}{2}x + \log\tau}{\sqrt{2 x}}.
\ee
The first derivative of $g(x)$ is
\be
g'(x) = \frac{1}{2\sqrt{2x}} \left( \frac{1}{2} - \frac{\log\tau}{x}\right).
\ee
Note that in the case in which $\log\tau \leqslant 0$ (or $\tau \leqslant 1$), then  for all $x \in \mathbb{R}^{+}$,  $g'(x) > 0$ and thus, $g$ is monotonically increasing with $x$. Since the complementary error function $\mathrm{erfc}$ is monotonically decreasing with its argument, the statement of Proposition~\ref{CorollaryTauSmall} follows and completes its proof.
In the case in which $\log\tau \geqslant 0$ (or $\tau > 1$), the solution to $g'(x) = 0$ is $x = 2 \log\tau$ and it corresponds to a minimum of the function $g$. The maximum of $\frac{1}{2}\mathrm{erfc}(g(x))$ occurs at the minimum of $g(x)$ given that $\mathrm{erfc}$ is monotonically decreasing with its argument. Hence, the maximum of $\sf{P_{ND}}(\av)$ occurs for the attack vectors satisfying:
\be
\label{EqQuadForm}
\av^{\sf{T}}\Sigmam_{Y\!Y}^{-1} \av = 2 \log\tau.
\ee
Solving for $\av$ in \eqref{EqQuadForm} yields \eqref{EqLowestProbAttack} and this completes the proof of Proposition~\ref{CorollaryTauBig}.
\end{proof}

The relevance of Proposition~\ref{CorollaryTauSmall} is that it states that when $\tau \leqslant 1$, any non-zero data-injection attack vector  possesses a non zero probability of being detected. Indeed, the highest probability $\sf{P_{ND}}(\av)$ of not being detected is guaranteed by  the null vector $\hbox{\av = (0, \ldots, 0)}$, i.e., there is no attack. Alternatively, when $\tau > 1$ it follows from Proposition~\ref{CorollaryTauBig} that there always exists a non-zero vector that possesses maximum probability of not being detected.
However, in both cases, it is clear that the corresponding data-injection vectors which induce the highest probability of not being detected are not necessarily the same that inflige the largest damage to the network, i.e., maximize the excess distortion.

From this point of view, the attacker faces the trade-off between maximizing the excess distortion and minimizing the probability of being detected. Thus, the attack construction can  be formulated as an optimization problem in which the solution $\av$ is a data-injection vector that maximizes the probability $\sf{P_{ND}}(\av)$ of not being detected at the same time that it induces a distortion $\| \cv \|_2^2 \geqslant {\sf D}_0$ into the estimate. In the case in which $\tau \leqslant 1$, it follows from Proposition~\ref{CorollaryTauSmall} and \eqref{eq:xv_a} that this problem can be formulated as the following optimization problem:
\be
\label{EqOP1}
\min_{\av\in\mathcal{A}} \av^{\sf{T}}\Sigmam_{Y\!Y}^{-1} \av \quad \textnormal{s.t.}\quad \av^{\sf{T}} \Sigmam_{Y\!Y}^{-1}\Hm  \Sigmam_{X\!X}^2 \Hm^{\sf{T}}\Sigmam_{Y\!Y}^{-1}\av \geq {\sf D}_0.
\ee
The solution to the optimization problem in \eqref{EqOP1} is given by the following theorem.
\begin{theorem}
\label{PropProb1}
Let $\Gm = \Sigmam_{Y\!Y}^{-\frac{1}{2}}\Hm  \Sigmam_{X\!X}^2 \Hm^{\sf{T}}\Sigmam_{Y\!Y}^{-\frac{1}{2}}$ have a singular value decomposition $\Gm  = \Um_{\Gm} \Sigmam_{\Gm} \Um_{\Gm}^{\sf{T}}$, with $\Um = \left( \uv_{\Gm,i}, \ldots, \uv_{\Gm,m} \right)$ a unitary matrix and $\Sigmam_{\Gm} = \diag\left( \lambda_{\Gm, 1}, \ldots, \lambda_{\Gm, m} \right)$ a diagonal matrix with $\lambda_{\Gm, 1} \geqslant \ldots \geqslant \lambda_{\Gm, m}$.
Then, if $\tau \leqslant 1$, the attack vector $\av$ that maximizes the probability of not being detected $\sf{P_{ND}}(\av)$ while inducing an excess distortion not less than $ {\sf D}_0$ is
\begin{align}
\label{eq:av_star_loco}
\av &\; = \; \pm\sqrt{\frac{ {\sf D}_0}{\lambda_{\Gm, 1}}}\Sigmam_{Y\!Y}^{\frac{1}{2}}\uv_{\Gm,1}.
\end{align}
Moreover, $\sf{P_{ND}}(\av) = \frac{1}{2} \mathrm{erfc}\left( \frac{\frac{{\sf D}_0}{2 \lambda_{\mathbf{G}, 1} } + \log\tau}{\sqrt{ \frac{2 {\sf D}_0}{\lambda_{\mathbf{G}, 1}}}}\right)$.
\end{theorem}

\begin{proof}
Consider the Lagrangian
\be
L(\av) = \av^{\sf{T}} \Sigmam_{Y\!Y}^{-1} \av - \gamma  \left(\av^{\sf{T}} \Sigmam_{Y\!Y}^{-1} \Hm \Sigmam_{X\!X}^{2} \Hm^{\sf{T}} \Sigmam_{Y\!Y}^{-1} \av  - {\sf D}_0  \right),
\ee
with $\gamma>0$ a Lagrangian multiplier.
Then, the necessary conditions for $\av$ to be a solution to the optimization problem \eqref{EqOP1} are:
\begin{align}
\label{EqCond1}
\nabla_{\av} L(\av) & = 2 \left( \Sigmam_{Y\!Y}^{-1}  - \gamma \Sigmam_{Y\!Y}^{-1} \Hm  \Sigmam_{X\!X}^2 \Hm^{\sf{T}} \Sigmam_{Y\!Y}^{-1} \right)  \av=0\\
\label{EqCond2}
\frac{\mathrm{d}}{\mathrm{d} \gamma} L(\av)  & =\av^{\sf{T}} \Sigmam_{Y\!Y}^{-1} \Hm  \Sigmam_{X\!X}^2 \Hm^{\sf{T}} \Sigmam_{Y\!Y}^{-1} \av -  {\sf D}_0=0.
\end{align}
Note that any
\begin{align}
\av_i & = \pm \sqrt{\frac{ {\sf D}_0}{\lambda_{\Gm, i}}}  \Sigmam_{Y\!Y}^{\frac{1}{2}} \uv_{\Gm,i} \mbox{ and }\\
\gamma_i & = \lambda_{\Gm,i}, \;\textnormal{with}\; 1\leqslant i \leqslant \mathrm{rank}\left( \Gm \right),
\end{align}
satisfy $\gamma_i > 0 $ and  conditions \eqref{EqCond1} and \eqref{EqCond2}.
Hence, the set of vectors that satisfy the necessary conditions to be a solution of \eqref{EqOP1} is
\be
\left\lbrace \av_{i} = \pm \sqrt{\frac{ {\sf D}_0}{\lambda_{\Gm, i}}} \Sigmam_{Y\!Y}^{\frac{1}{2}}\uv_{\Gm,i}: 1 \leqslant i \leqslant \mathrm{rank}\left( \Gm \right) \right\rbrace.
\ee
More importantly, any vector $\av \neq \av_{i}$, with $1 \leqslant i \leqslant \mathrm{rank}\left( \Gm \right)$, does not satisfy the necessary conditions.
Moreover,
\be
\av_{i}^{\sf{T}} \Sigmam_{Y\!Y} ^{-1} \av_{i} =  \frac{{\sf D}_0}{\lambda_{\Gm, i}} \geqslant  \frac{{\sf D}_0}{\lambda_{\Gm, 1}}.
\ee
Therefore, $\av = \pm  \sqrt{ \frac{{\sf D}_0}{\lambda_{\Gm, 1}}} \Sigmam_{Y\!Y}^{\frac{1}{2}}\uv_{\Gm,1}$ are the unique solutions to \eqref{EqOP1}.
This completes the proof.
\end{proof}

Interestingly, the construction of the data-injection attack $\av$ in \eqref{eq:av_star_loco} does not require the exact knowledge of $\tau$. That is, only knowing that $\tau \leqslant 1$ is enough to build the data-injection attack that has the highest  probability of not being detected and induces a distortion of at least $\sf{D}_{0}$.

In the case in which $\tau > 1$, it is also possible to find the data-injection attack vector that induces a distortion not less than $\sf{D}_{0}$ and the maximum probability of not being detected. Such a vector is the solution to the following optimization problem.
\be
\label{EqOP2}
\min_{\av\in\mathcal{A}} \frac{\frac{1}{2} \av^{\sf{T}}\Sigmam_{Y\!Y}^{-1} \av + \log\tau }{\sqrt{2 \av^{\sf{T}}\Sigmam_{Y\!Y}^{-1} \av }}  \;\; \textnormal{s.t.} \;\;  \av^{\sf{T}} \Sigmam_{Y\!Y}^{-1}\Hm  \Sigmam_{X\!X}^2 \Hm^{\sf{T}}\Sigmam_{Y\!Y}^{-1}\av \geq {\sf D}_0.
\ee
The solution to the optimization problem in \eqref{EqOP2} is given by the following theorem.
\begin{theorem}\label{PropProb2}
Let $\Gm = \Sigmam_{Y\!Y}^{-\frac{1}{2}}\Hm  \Sigmam_{X\!X}^2 \Hm^{\sf{T}}\Sigmam_{Y\!Y}^{-\frac{1}{2}}$ have a singular value decomposition $\Gm  = \Um_{\mathbf{G}} \Sigmam_{\Gm} \Um_{\Gm}^{\sf{T}}$, with $\Um_{\Gm} = \left( \uv_{\Gm,i}, \ldots, \uv_{\Gm,m} \right)$ a unitary matrix and $\Sigmam_{\Gm} = \diag\left( \lambda_{\Gm, 1}, \ldots, \lambda_{\Gm, m} \right)$ a diagonal matrix with $\lambda_{\Gm, 1} \geqslant \ldots \geqslant \lambda_{\Gm, m}$.
Then, when $\tau > 1$, the attack vector $\av$ that maximizes the probability of not being detected $\sf{P_{ND}}(\av)$ while producing an excess distortion not less than $ {\sf D}_0$ is
\begin{align}
\nonumber
\av & = \left \lbrace
\begin{array}{ll}
{\tiny\pm} \sqrt{\frac{ {\sf D}_0}{\lambda_{\Gm, k^*}}}\Sigmam_{Y\!Y}^{\frac{1}{2}}\uv_{\Gm,k^*} & \text{ if } \frac{ {\sf D}_0}{2\log\tau  \lambda_{\Gm, \mathrm{rank} \, \Gm} } \geqslant 1,\\
{\tiny \pm} \sqrt{2 \log\tau}\Sigmam_{Y\!Y}^{\frac{1}{2}}\uv_{\Gm,1} & \text{ if } \frac{ {\sf D}_0}{2\log\tau  \lambda_{\Gm, \mathrm{rank} \, \Gm} }  < 1
\end{array}
\right.
\end{align}
with
\be\label{EqKstar}
k^* = \displaystyle\mathrm{arg}\min_{k \in \lbrace 1, \ldots, \mathrm{rank} \Gm \rbrace : \frac{\sf{D}_0}{\lambda_{\Gm, k}} > 2 \log (\tau)} \frac{\sf{D}_0}{\lambda_{\Gm, k}} .
\ee
\end{theorem}
\begin{proof}
The structure of the proof of Theorem~\ref{PropProb2} is similar to the proof of Theorem \ref{PropProb1} and is omitted in this chapter. A complete proof can be found in \cite{EPKP_inria_15}.
\end{proof}

\subsection{Attacks with Maximum Distortion}

In the previous subsection, the attacker constructs its data-injection vector $\av$ aiming to maximize the probability of non-detection $\sf{P_{ND}}(\av)$ while guaranteeing a minimum distortion. However, this problem has a dual in which the objective is to maximize the distortion $\av^{\sf{T}} \Sigmam_{Y\!Y}^{-1}\Hm  \Sigmam_{X\!X}^2 \Hm^{\sf{T}}\Sigmam_{Y\!Y}^{-1}\av$ while guaranteeing that the probability of not being detected remains always larger than a given threshold $L'_0 \in [0,\frac{1}{2}]$.
This problem can be formulated as the following optimization problem:
 \be
 \label{EqOP3}
\max_{\av\in\mathcal{A}}  \av^{\sf{T}} \Sigmam_{Y\!Y}^{-1} \Hm \Sigmam_{X\!X}^2 \Hm^{\sf{T}} \Sigmam_{Y\!Y}^{-1} \av \; \textnormal{ s.t. } \; \frac{\frac{1}{2} \av^{\sf{T}}\Sigmam_{Y\!Y}^{-1} \av + \log\tau }{\sqrt{2 \av^{\sf{T}}\Sigmam_{Y\!Y}^{-1} \av }} \leq L_0,
\ee
with $L_0 = \mathrm{erfc}^{-1} \left(2 L'_0 \right) \in [0, \infty)$.

The solution to the optimization problem in \eqref{EqOP3} is given by the following theorem.

\begin{theorem}
\label{Th:prob4}
Let the matrix $\Gm = \Sigmam_{Y\!Y}^{-\frac{1}{2}} \Hm  \Sigmam_{X\!X}^2 \Hm^{\sf{T}} \Sigmam_{Y\!Y}^{-\frac{1}{2}}$ have a singular value decomposition $\Um_{\Gm} \Sigmam_{\Gm} \Um_{\Gm}^{\sf{T}}$, with $\Um = \left( \uv_{\Gm,i}, \ldots, \uv_{\Gm,m} \right)$ a unitary matrix and $\Sigmam_{\Gm} = \diag\left( \lambda_{\Gm, 1}, \ldots, \lambda_{\Gm, m} \right)$ a diagonal matrix with $\lambda_{\Gm, 1} \geqslant \ldots \geqslant \lambda_{\Gm, m}$.
Then, the attack vector $\av$ that maximizes the excess distortion $\av^{\sf{T}} \Sigmam_{Y\!Y}^{-\frac{1}{2}}\Gm \Sigmam_{Y\!Y}^{-\frac{1}{2}}\av$ with a probability of not being detected that does not go below $L_0 \in [0, \frac{1}{2}]$ is
\begin{align}
\label{eq:av_star}
\av & = \pm \left(\sqrt{2} L_0 + \sqrt{2 L_0^2 - 2 \log\tau}\right) \Sigmam_{Y\!Y}^{\frac{1}{2}}\uv_{\Gm,1},
\end{align}
when a solution exists.
\end{theorem}
\begin{proof}
The structure of the proof of Theorem~\ref{Th:prob4} is similar to the proof of Theorem \ref{PropProb1} and is omitted in this chapter. A complete proof can be found in \cite{EPKP_inria_15}.
\end{proof}

\section{Decentralized Deterministic Attacks}\label{SecDA}
\label{sec:d_attacks}
Let  $\Kc = \lbrace 1, \ldots, K \rbrace$ be the set of attackers that can potentially perform a data injection attack on the network, e.g., a decentralized vector attack. Let also $\Cc_k \in \{1,2,\ldots,m\}$ be the set of sensors that attacker $k \in \Kc$ can control. Assume that $\Cc_1, \ldots, \Cc_K$ are proper sets and form a partition of the set $\Mc$ of all sensors.
The set $\Ac_k$ of data attack vectors $\av_k = (a_{k,1}, a_{k,2}, \ldots, a_{k,m})$ that can be injected into the network by attacker $k \in \Kc$ is of the form
\be
\label{EqSetAk}
\Ac_k=\lbrace \av_k\in\mathbb{R}^m: \av_{k,j}=0\;\text{for all}\; j\notin \Cc_k, \av_k^{\sf{T}}\av_k\leq E_k \rbrace.
\ee
The constant $E_k < \infty$ represents the energy budget of attacker $k$.
Let the set of all possible sums of the elements of $\Ac_i$ and $\Ac_j$ be denoted by $\Ac_i \oplus \Ac_j$. That is, for all $\av \in \Ac_i \oplus \Ac_j$, there exists a pair of vectors $(\av_i, \av_j) \in \mathcal{A}_i \times \mathcal{A}_j$ such that $\av = \av_i + \av_j$.
Using this notation, let the set of all possible data-injection attacks be denoted by
\begin{align}
\Ac & =   \Ac_{1} \oplus \Ac_{2} \oplus \ldots \oplus \Ac_{K} ,
\end{align}
and the set of complementary data-injection attacks with respect to attacker $k$ be denoted by
\begin{align}
\Ac_{-k} & =   \Ac_{1} \oplus \ldots \oplus \Ac_{k-1} \oplus \Ac_{k+1} \oplus \ldots \oplus \Ac_{K}.
\end{align}
Given the individual data injection vectors $\av_i \in \Ac_i$, with $i \in \lbrace 1, \ldots, K \rbrace$, the global  attack vector $\av$ is
\be
\label{eq:va}
\av = \displaystyle\sum_{i = 1}^{K} \av_k \in \Ac.
\ee
The aim of attacker $k$ is to corrupt the measurements obtained by the set of meters $\Cc_k$ by injecting an error vector $\av_k \in \mathcal{A}_k$ that maximizes the damage to the network, e.g., the excess distortion, while avoiding the detection of the global data-injection vector $\av$.
Clearly, all attackers have the same interest but they control different sets of measurements, i.e., $\Cc_i \neq \Cc_k$, for a any pair $(i,k) \in \Kc^2$.
For modeling this behavior, attackers use the utility function $\phi: \mathbb{R}^{m}  \rightarrow  \mathbb{R}$, to determine whether a data-injection vector $\av_k \in \mathcal{A}_k$ is more beneficial than another $\av'_k \in \mathcal{A}_k$ given the complementary attack vector
\be
\av_{-k} = \displaystyle\sum_{i \in \lbrace 1, \ldots, K \rbrace \setminus \lbrace k \rbrace}  \av_i  \in \Ac_{-k}
\ee
adopted by all the other attackers. The function $\phi$ is chosen considering the fact that an attack is said to be successful if it induces a non-zero distortion and it is not detected. Alternatively, if the attack is detected no damage is induced into the network as the operator discards the measurements and no estimation is performed. Hence, given a global attack $\av$, the distortion induced into the measurements is $\mathds{1}_{\left\lbrace L(Y_{a}^{m}, \av) > \tau \right\rbrace} \xv_a^{\sf{T}}\xv_a$. However, attackers are not able to know the exact state of the network $\xv$ and the realization of the noise $\zv$ before launching the attack. Thus, it appears natural to exploit the knowledge of the first and second moments of both the state variables $\xv$ and noise $\zv$ and consider as a metric the expected distortion $\phi(\av)$ that can be induced by the attack vector $\av$:
\begin{align}
\label{eq:g_utility}
\phi(\av) & = \sf{E}\left[ \left( \mathds{1}_{\left\lbrace L(Y_{a}^{m}, \av) > \tau \right\rbrace} \right) \cv^{\sf{T}}\cv\right],\\
& = \sf{P_{ND}}(\av) \; \av^{\sf{T}} \Sigmam_{Y\!Y}^{-1}\Hm  \Sigmam_{X\!X}^2 \Hm^{\sf{T}}\Sigmam_{Y\!Y}^{-1}\av,
\end{align}
where $\cv$ is in (\ref{eq:xv_a}) and the expectation is taken over the distribution of state variables $X^{n}$ and the noise $Z^{m}$. Note that under this assumptions of global knowledge, this model considers the worst case scenario for the network operator. Indeed, the result presented in this section corresponds to the case in which the attackers inflict the most harm onto the state estimator.

\subsection{Game Formulation}
The benefit $\phi(\av)$ obtained by attacker $k$ does not only depend on its own data-injection vector $\av_k$, but also on the data-injection vectors $\av_{-k}$ of all the other attackers. This becomes clear from the construction of the global data-injection vector $\av$ in \eqref{eq:va}, the excess distortion $\xv_{a}$ in \eqref{eq:xv_a} and the probability of not being detected $\sf{P_{ND}}(\av)$ in  \eqref{EqProbND}.
Therefore, the interaction of all attackers in the network can be described by a game in normal form
\be
\label{EqGame}
\GameNF.
\ee
Each attacker is a player in the game $\gameNF$ and it is identified by an index from the set $\Kc$.
The actions player $k$ might adopt are data-injection vectors $\av_k$ in the set $\Ac_k$ in \eqref{EqSetAk}.
The underlying assumption in the following of this section is that, given a vector of data-injection attacks $\av_{-k}$, player $k$ aims to adopt a data-injection vector $\av_k$ such that the expected excess distortion $\phi(\av_k+ \av_{-k})$ is maximized.
That is,
\be
\av_k \in \BR_k\left( \av_{-k} \right),
\ee
where the correspondence $\BR_k: \Ac_{-k} \rightarrow 2^{\Ac_{k}}$ is the best response correspondence, i.e.,
\be
\BR_{k}\left( \av_{-k} \right) = \arg\max_{\av_k \in \Ac_k} \phi \left( \av_k + \av_{-k} \right).
\ee
The notation $2^{\Ac_{k}}$ represents the set of all possible subsets of $\Ac_{k}$.
Note that $\mathrm{BR}_{k}\left( \av_{-k} \right) \subseteq \Ac_k$ is the set of data-injection attack vectors that are optimal given that the other attackers have adopted the data-injection vector $\av_{-k}$.  In this setting, each attacker tampers with a subset $\Cc_k$ of all sensors $\Cc = \{ 1,2, \ldots, m\}$, as opposed to the centralized case in which there exists a single attacker that is able to tampers with all sensors in $\Cc$.

A game solution that is particularly relevant for this analysis is the NE \cite{Nash-PNAS-1950}.

\begin{definition}[Nash Equilibrium] \label{def:NE}
The data-injection vector $\av$  is an NE  of the game $\gameNF$ if and only if it is a solution of the fix point equation
\begin{align}\label{eq:defNE}
\av = \BR\left(\av\right),
\end{align}
with $\BR: \Ac \rightarrow  2^{\Ac}$ being the global best-response correspondence, i.e.,
\be
\BR\left(\av \right) = \BR_1\left( \av_{-1}\right) + \ldots + \BR_K\left( \av_{-K}\right).
\ee
\end{definition}
Essentially, at an NE, attackers obtain the maximum benefit given the data-injection vector adopted by all the other attackers.
This implies that an NE is an operating point at which attackers achieve the highest expected distortion induced over the measurements.
More importantly,  any unilateral deviation from an equilibrium data-injection vector $\av$ does not lead to an improvement of the average excess distortion.
Note that this formulation does not say anything about the exact distortion induced by an attack but the average distortion. This is mainly because the attack is chosen under the uncertainty of the state vector $X^{n}$ and the noise term $Z^{m}$.

The following proposition highlights an important property of the game $\gameNF$ in  \eqref{EqGame}.

\begin{proposition}\label{PropPotential} The game $\gameNF$ in \eqref{EqGame} is a potential game.
\end{proposition}
\begin{proof}
The proof follows immediately from the observation that all the players have the same utility function $\phi$ \cite{Monderer-GEB-1996}. Thus, the function $\phi$ is a potential of the game $\gameNF$ in \eqref{EqGame} and any maximum of the potential function is an NE of the game $\gameNF$.
\end{proof}

In general, potential games \cite{Monderer-GEB-1996} possess numerous properties that are inherited by the game $\gameNF$ in \eqref{EqGame}. These properties are detailed by the following  propositions

\begin{proposition}\label{PropExistenceNE}
The game $\gameNF$ possesses at least one NE.
\end{proposition}
\begin{proof}
Note that $\phi$ is continuous in $\Ac$ and $\Ac$ is a convex and closed set; therefore, there always exists a maximum of the potential function $\phi$ in $\Ac$. Finally from Lemma $4.3$ in \cite{Monderer-GEB-1996}, it follows that such a maximum corresponds to an NE.
\end{proof}

\subsection{Achievability of an NE}

The attackers are said to play a sequential best response dynamic (BRD) if the attackers can sequentially decide their own data-injection vector $\av_k$ from their sets of best responses following a round-robin (increasing) order. Denote by $\av_{k}^{(t)} \in \Ac$ the  choice of attacker $k$ during round $t \in \mathbb{N}$ and assume that attackers are able to observe all the other attackers' data-injection vectors.
Under these assumptions, the BRD can be defined as follows.
\begin{definition}[Best Response Dynamics]
The players of the game $\gameNF$ are said to play best response dynamics if there exists a round-robin order of the elements of $\Kc$ in which at each round  $t \in \mathbb{N}$, the following holds:
\be
\av_k^{(t)} \in \BR_{k} \left( \av_{1}^{(t)}+ \ldots +\av_{k-1}^{(t)} + \av_{k+1}^{(t-1)} + \ldots + \av_{K}^{(t-1)} \right).
\ee
\end{definition}
From the properties of potential games (Lemma $4.2$ in \cite{Monderer-GEB-1996}), the following proposition follows.
\begin{lemma}[Achievability of NE attacks]\label{PropBRD}
Any BRD in the game $\gameNF$ converges to a data-injection attack vector that is an NE.
\end{lemma}

 The relevance of Lemma~\ref{PropBRD} is that it establishes that if attackers can communicate in at least a round-robin fashion, they are always able to attack the network with a data-injection vector that maximizes the average excess distortion. Note that there might exists several NEs (local maxima of $\phi$) and there is no guarantee that attackers will converge to the best NE, i.e., a global maximum of $\phi$. It is important to note that under the assumption that there exists a unique maximum, which is not the case for the game
$\gameNF$ (see Theorem \ref{PropCardinalityNE}), all attackers are able to calculate such a global maximum and no communications is required among the attackers. Nonetheless, the game $\gameNF$ always possesses at least two NEs, which enforces the use of a sequential BRD to converge to an NE.

\subsection{Cardinality of the set of NEs}

Let $\Ac_{\mathrm{NE}}$ be the set of all data-injection attacks that form NEs. The following theorem bounds the number of NEs in the game.

\begin{theorem}\label{PropCardinalityNE}
The cardinality of the set $\Ac_{\mathrm{NE}}$ of NE of the game $\gameNF$ satisfies
\be
2\leqslant|\Ac_{\mathrm{NE}}| \leqslant C\cdot\textnormal{rank}(\Hm)
\ee
where $C<\infty$ is a constant that depends on $\tau$.
\end{theorem}
\begin{proof}
The lower bound follows from the symmetry of the utility function given in (\ref{eq:g_utility}), i.e. $\phi(\av)=\phi(-\av)$, and the existence of at least one NE claimed in Proposition \ref{PropExistenceNE}.

To prove the upper bound the number of stationary points of the utility function is evaluated. This is equivalent to the cardinality of the set
\be
\label{eq:set_SP}
\mathcal{S}=\lbrace\av\in\mathbb{R}^m:\nabla_{\av}\phi(\av)=\mathbf{0}\rbrace,
\ee
which satisfies $\mathcal{A}_{NE}\subseteq\mathcal{S}$.
Calculating the gradient with respect to the attack vector yields
\be
\label{eq:util_gradient}
\nabla_{\av}\phi(\av)=\left(\alpha(\av) \Mm^{\sf T}\Mm-\beta(\av)\Sigmam_{Y\!Y}^{-1}\right) \av,
\ee
where
\be
\alpha(\av)\eqdef\text{erfc}\left(\frac{1}{\sqrt{2}}\frac{\frac{1}{2}\av^{\sf T}\Sigmam_{Y\!Y}^{-1}\av+\log\tau}{\left(\av^{\sf T}\Sigmam_{Y\!Y}^{-1}\av\right)^{\frac{1}{2}}}\right)
\ee
and
\begin{align}
\beta(\av)&\;\eqdef\;\frac{\av^{\sf T}\Mm^{\sf T}\Mm\av}{\sqrt{2\pi} \av^{\sf T}\Sigmam_{Y\!Y}^{-1}\av}\left(\frac{1}{2}-\frac{\log\tau}{\av^{\sf T}\Sigmam_{Y\!Y}^{-1}\av}\right)\exp\left(-\left(\frac{1}{\sqrt{2}}\frac{\frac{1}{2}\av^{\sf T}\Sigmam_{Y\!Y}^{-1}\av+\log\tau}{\left(\av^{\sf T}\Sigmam_{Y\!Y}^{-1}\av\right)^{\frac{1}{2}}}\right)^2\right).
\end{align}

Define $\delta(\av)\eqdef\frac{\beta(\av)}{\alpha(\av)}$ and note that combining  (\ref{EqMMSEm}) with  (\ref{eq:util_gradient}) gives the following condition for the stationary points:
\be
\label{eq:grad_cond_1}
\left(\Hm\Sigmam_{X\!X}^2\Hm^{\sf T}\Sigmam_{Y\!Y}^{-1}-\delta(\av)\mathbf{I}\right)\av=\mathbf{0}.
\ee
Note that the number of linearly independent attack vectors that are a solution of the linear system in (\ref{eq:grad_cond_1}) is given by
\begin{align}
\label{eq:rank}
R&\;\eqdef\;\text{rank}\left(\Hm\Sigmam_{X\!X}^2\Hm^{\sf T}\Sigmam_{Y\!Y}^{-1}\right)\\
\label{eq:rank_2}
&\;=\;\text{rank}\left(\Hm\right).
\end{align}
where (\ref{eq:rank_2}) follows from the fact that $\Sigmam_{X\!X}$ and $\Sigmam_{Y\!Y}$ are positive definite. Define the eigenvalue decomposition
\be
\Sigmam_{Y\!Y}^{-\frac{1}{2}}\Hm\Sigmam_{X\!X}^2\Hm^{\sf T}\Sigmam_{Y\!Y}^{-\frac{1}{2}}=\Um\Lambdam\Um^{\sf T}
\ee
where $\Lambdam$ is a diagonal matrix containing the ordered eigenvalues $\lbrace \lambda_i\rbrace_{i=1}^m$ matching the order of of the eigenvectors in $\Um$. As a result of (\ref{eq:rank}) there are $\textcolor{black}{r}$ eigenvalues, $\lambda_k$, which are different from zero and $m-\textcolor{black}{r}$ diagonal elements of $\Lambdam$ which are zero. Combining this decomposition with some algebraic manipulation, the condition for stationary points in (\ref{eq:grad_cond_1}) can be recast as
\be
\label{eq:grad_cond_2}
\Sigmam_{Y\!Y}^{-\frac{1}{2}}\Um\left(\Lambdam-\delta(\av)\mathbf{I}\right)\Um^{\sf T}\Sigmam_{Y\!Y}^{-\frac{1}{2}}\av=\mathbf{0}.
\ee
Let $w\in\mathbb{R}$ be a scaling parameter and observe that the attack vectors that satisfy $\av=w\Sigmam_{Y\!Y}^{\frac{1}{2}} \Um\ev_k$ and $\delta(\av)=\lambda_k$ for $k=1,\ldots, \textcolor{black}{r}$ are solutions of (\ref{eq:grad_cond_2}). Note that the critical points associated to zero eigenvalues are not NE. Indeed, the eigenvectors associated to zero eigenvalues yield zero utility. Since the utility function is strictly positive, these critical points are minima of the utility function and can be discarded when counting the number of NE. Therefore, the set in (\ref{eq:set_SP}) can be rewritten based on the condition in (\ref{eq:grad_cond_2}) as
\be
\mathcal{S}=\bigcup_{k=1}^R\mathcal{S}_k,
\ee
where
\be
\label{eq:set_SP_2}
\mathcal{S}_k=\lbrace\av\in\mathbb{R}^m:\av=w\Sigmam_{Y\!Y}^{\frac{1}{2}} \Um\ev_k\;\textnormal{and}\;\delta(\av)=\lambda_k\rbrace.
\ee
There are $\textcolor{black}{r}$ linearly independent solutions of (\ref{eq:grad_cond_2}) but for each linearly independent solution there can be several scaling parameters, $w$, which satisfy $\delta(\av)=\lambda_k$. For that reason, $|\mathcal{S}_k|$ is determined by the number of scaling parameters that satisfy $\delta(\av)=\lambda_k$. To that end, define $\delta':\mathbb{R}\rightarrow \mathbb{R}$ as $\delta'(w)\eqdef\delta(w\Sigmam_{Y\!Y}^{\frac{1}{2}} \Um\ev_k)$. It is easy to check that $\delta'(w)=\lambda_k$ has a finite number of solutions for $k=1,\ldots, \textcolor{black}{r}$. Hence, for all $k$ there exists a constant $C_k$ such that $|\mathcal{S}_k|\leq C_k$ which yields the upper bound
\be
|\mathcal{S}|\leq\sum_{i=1}^R|\mathcal{S}_k|\leq\sum_{i=1}^R C_k\leq \max_{k}C_k R.
\ee
Noticing that the there is a finite number of solutions of $\delta'(w)=\lambda_k$ and that they depend only on $\tau$ yields the upper  bound.

\end{proof}

\section{Information-Theoretic Attacks} \label{System_Model}

Modern sensing infrastructure is moving toward increasing the number of measurements that the operator acquires, e.g. phasor measurement units exhibit temporal resolutions in the order of miliseconds while supervisory control and data acquisition (SCADA) systems traditionally operate with a temporal resolution in the order of seconds. As a result, attack constructions that do not change within the same temporal scale at which measurements are reported do not exploit all the \emph{degrees of freedom} that are available to the attacker. Indeed, an attacker can choose to change the attack vector with every measurement vector that is reported to the network operator. However, the deterministic attack construction changes when the Jacobian measurement matrix changes, i.e. with the operation point of the system. Thus, in the deterministic attack case, the attack construction changes at the same rate that the Jacobian measurement matrix changes and, therefore, the dynamics of the state variables define the update cadency of the attack vector.

In this section, we study the case in which the attacker constructs the attack vector as a random process that corrupts the measurements. By endowing the attack vector with a probabilistic structure we provide the attacker with an attack construction strategy that generates attack vector realizations over time and that achieve a determined objective on average. In view of this, the task of the attacker in this case is to devise the optimal distribution for the attack vectors. In the following, we pose the attack construction problem within an information-theoretic framework and characterize the attacks that simultaneously minimize the mutual information and the probability of detection.

\subsection{Random Attack Model}
%

We consider an additive attack model as in (\ref{eq:obs_mod_det_a}) but with the distinction that the attack is a random process. The resulting vector of compromised measurements is given by
\begin{align}
\label{eq:measurement_model}
Y^{m}_{A} = \Hm X^{m} + Z^{m} + A^{m},
\end{align}
where $A^{m} \in \RR^{m}$ is the vector of random variables introduced by the attacker and $Y^{m}_{A} \in \RR^{m } $ is the vector containing the compromised measurements. The attack vector of random variables is described by the distribution $P_{A^m}$ which is the determined by the attacker. We assume that the attacker has no access to the realizations of the state variables, and therefore, it holds that $P_{A^mX^n}=P_{A^m}P_{X^n}$ where $P_{A^m X^{n}}$ denotes the joint distribution of $A^{m}$ and $X^{n}$.

Similarly to the deterministic attack case, we adopt a multivariate Gaussian framework for the state variables such that $X^n\thicksim\Nc(\mathbf{0},\Sigmam_{X\!X})$. Moreover, we limit the attack vector distribution to the set of zero-mean multivariate Gaussian distributions, i.e.   $A^{m} \sim  \Nc (\zerov,\Sigmam_{A\!A})$ where $\Sigmam_{A\!A}\in \Sc^{m}_{+}$ is the covariance matrix of the attack distribution. {The rationale for choosing a Gaussian distribution for the attack vector follows from the fact that for the measurement model in (\ref{eq:measurement_model}) the additive attack distribution that minimizes the mutual information between the vector of state variables and the compromised measurements is Gaussian \cite{SA_TIT_13}. As we will see later, minimizing this mutual information is central to the proposed information-theoretic attack construction and indeed one of the objectives of the attacker. }
Because of the Gaussianity of the attack distribution, the vector of compromised measurements is distributed as
\begin{align}
  Y_{A}^{m} \sim \Nc(\zerov,\Sigmam_{Y_{A}\!Y_{A}}),
\end{align}
where $\Sigmam_{Y_{A}\!Y_{A}} = \Hm\Sigmam_{X\!X}\Hm^{\sf T} + \sigma^{2}\Id + \Sigmam_{A\!A} $ is the covariance matrix of the distribution of the compromised measurements. Note that while in the case of deterministic attacks the effect of the attack vector was captured by shifting the mean of the measurement vector, in the random attack case the attack changes the structure of the second order moments of the measurements. 
Interestingly, the Gaussian attack construction implies that knowledge of the second order moments of the state variables and the variance of the AWGN introduced by the measurement process suffices to construct the attack. This assumption significantly reduces the difficulty of the attack construction.


The operator of the power system makes use of the acquired measurements to detect the attack.
The detection problem is cast as a hypothesis testing problem with hypotheses
\begin{align}
\Hc_{0}:  \ & Y^{m} \sim \Nc(\zerov,\Sigmam_{Y\!Y}), \quad \text{versus}  \nonumber \\
\Hc_{1}:  \ & Y^{m} \sim \Nc(\zerov,\Sigmam_{Y_{A}\!Y_{A}}).
\end{align}
The null hypothesis $\Hc_{0}$ describes the case in which the power system is not compromised, while the alternative hypothesis $\Hc_{1}$ describes the case in which the power system is under attack.

Two types of error are considered in hypothesis testing problems,
Type \RNum{1} error
is the probability of a ``true negative'' event;
and Type \RNum{2} error
is the probability of a ``false alarm'' event.
The Neyman-Pearson lemma \cite{neyman_problem_1992} states that for a fixed probability of Type \RNum{1} error,
the likelihood ratio test (LRT) achieves the minimum Type \RNum{2} error when compared with any other test with an equal or smaller Type \RNum{1} error.
Consequently, the LRT is chosen to decide between $\Hc_{0}$ and $\Hc_{1}$ based on the available measurements.
The LRT between $\mathcal{H}_{0}$ and $\mathcal{H}_{1}$ takes following form:
\begin{equation}\label{LHRT}
L(\yv) \eqdef \frac{f_{Y^{m}_{A}}(\yv)}{f_{Y^{m}}(\yv)} \ \LRT{\Hc_{1}}{\Hc_{0}} \ \tau,
\end{equation}
where $\yv \in \RR^{m}$ is a realization of the vector of random variables modelling the measurements, $f_{Y_A^m}$ and  $f_{Y^m}$ denote the probability density functions (p.d.f.'s) of $Y_A^m$ and  $Y^m$, respectively, and $\tau$ is the decision threshold set by the operator to meet the false alarm constraint. 

\subsection{Information-Theoretic Setting}
The aim of the attacker is twofold. Firstly, it aims to disrupt the state estimation process by corrupting the measurements in such a way that the network operator acquires the least amount of knowledge about the state of the system. Secondly, the attacker aspires to remain stealthy and corrupt the measurements without being detected by the network operator. In the following we propose to information-theoretic measures that provide quantitative metrics for the objectives of the attacker.

The data-integrity of the measurements is measured in terms of the mutual information between the state variables and the measurements. The mutual information between two random variables is a measure of the amount of information that each random variable contains about the other random variable.  By adding the attack vector to the measurements the attacker aims to reduce the mutual information which ultimately results in a loss of information about the state by the network operator. Specifically, the attacker aims to minimize $I(X^{n};Y_{A}^{m})$. In view of this, it seems reasonable to consider a Gaussian distribution for the attack vector as the minimum mutual information for the observation model in (\ref{eq:obs_mod_det_a}) is achieved by additive Gaussian noise. 

The probability of attack detection is determined by the detection threshold $\tau$ set by the operator for the LRT and the distribution induced by the attack on the vector of compromised measurements. An analytical expression of the probability of attack detection can be described in closed-form as a function of the distributions describing the measurements under both hypotheses. However, the expression is involved in general and it is not straightforward to incorporate it into an analytical formulation of the attack construction. For that reason, we instead consider the asymptotic performance of the LRT to evaluate the detection performance of the operator.
The Chernoff-Stein lemma \cite{cover_elements_2012} characterizes
the asymptotic exponent of the probability of detection when the number of observations of measurement vectors grows to infinity. In our setting, the Chernoff-Stein lemma states that for any LRT and $\epsilon  \in (0,1/2)$, it holds that
\begin{align}
\label{eq:Chernoff-Stein}
\lim_{T \to \infty} \frac{1}{\textcolor{black}{T}} \log \beta_{T}^{\epsilon} = -D(P_{Y^{m}_{A}}||P_{Y^{m}}),
\end{align} 
where $D(\cdot ||\cdot)$ is the Kullback-Leibler (KL) divergence, $\beta_{T}^{\epsilon}$ is the minimum Type II error such that the Type I error $\alpha$ satisfies $\alpha < \epsilon$, and $\textcolor{black}{T}$ is the number of $m$-dimensional measurement vectors that are available for the LRT detection procedure.
As a result, minimizing the asymptotic probability of false alarm given an upper bound on the probability of misdetection is equivalent to minimizing $D(P_{Y^{m}_{A}}||P_{Y^{m}})$, where $P_{Y_A^m}$ and  $P_{Y^m}$ denote the probability distributions of $Y_A^m$ and  $Y^m$, respectively. 

The purpose of the attacker is to disrupt the normal state estimation procedure by
minimizing the information that the operator acquires about the state variables, while guaranteeing that the probability of attack detection is sufficiently small, and therefore, remain stealthy.

%
\subsection{Generalized Stealth Attacks}
\label{SEC:ITA}

When the two information-theoretic objectives are considered by the attacker, in \cite{Sun_information-theoretic_2017}, a stealthy attack construction is proposed by combining  two objectives in one cost function, i.e.,
\begin{align}
\label{Equ:Steallth_Obj}
 I(X^{n};Y^{m}_{A}) \hspace{-0.1em} +  \hspace{-0.2em} D( P_{{Y}^{m}_{A}}||P_{Y^{m}}) \hspace{-0.2em} =  \hspace{-0.2em} D( P_{X^{n}Y_{A}^{m}}||P_{X^{n}}P_{Y^{m}}),\hspace{-0.2em} 
\end{align}
where $P_{X^{n}Y_{A}^{m}}$ is the joint distribution of $X^{n}$ and $Y_{A}^{m}$.
The resulting optimization problem to construct the attack is given by
\begin{align}
\label{Stealth_Obj}
\underset{A^{m}}{\text{min}} \ D( P_{X^{n}Y_{A}^{m}}||P_{X^{n}}P_{Y^{m}}).
\end{align}
Therein, it is shown that (\ref{Stealth_Obj}) is a convex optimization problem and the covariance matrix of the optimal Gaussian attack is $\Sigmam_{A\!A} = \Hm\Sigmam_{X\!X}\Hm^{\sf T}$.
%
However, numerical simulations on IEEE test system show that the attack construction proposed in the preceding text yields large values of probability of detection in practical settings.

To control the probability of attack detection of the attack, the preceding construction is generalized in \cite{Sun_Stealth_2020} by introducing a parameter that weights the detection term in the cost function.
The resulting optimization problem is given by
\begin{align}\label{Equ:WeightSum}
\underset{A^{m}}{\text{min}} \  I(X^{n};Y^{m}_{A})  +  \lambda D( P_{{Y}^{m}_{A}}||P_{Y^{m}}),  
\end{align}
where $\lambda\geq1$ governs the weight given to each objective in the cost function. It is interesting to note that for the case in which $\lambda=1$ the proposed cost function boils down to the effective secrecy proposed in \cite{hou_effective_2014} and the attack construction in (\ref{Equ:WeightSum}) coincides with that in \cite{Sun_information-theoretic_2017}.
For $\lambda>1$, the attacker adopts a conservative approach and prioritizes remaining undetected over minimizing the amount of information acquired by the operator.
By increasing the value of $\lambda$ the attacker decreases the probability of detection at the expense of increasing the amount of information acquired by the operator using the measurements.


The attack construction in (\ref{Equ:WeightSum}) is formulated in a general setting. The following propositions particularize the KL divergence and MI to our multivariate Gaussian setting.
\begin{proposition}{ \rm \cite{cover_elements_2012}}
The KL divergence between $m$-dimensional multivariate Gaussian distributions $ \Nc(\zerov,\Sigmam_{Y_{A}\!Y_{A}}) $ and $\Nc(\zerov,\Sigmam_{Y\!Y}) $ is given by
\begin{align}\label{KL}
D( P_{Y^{m}_{A}}||P_{Y^{m}}) \hspace{-0.1em} = \hspace{-0.1em} \frac{1}{2} \left( \log\frac{|\Sigmam_{Y\!Y}|}{| \Sigmam_{Y_{A}\!Y_{A}}|} -\hspace{-0.1em} m \hspace{-0.1em}+\text{\rm tr}\left(\Sigmam_{Y\!Y}^{-\!1}\Sigmam_{Y_{A}\!Y_{A}}\right)\hspace{-0.2em} \right)\hspace{-0.2em}.
\end{align}
\end{proposition}
\begin{proposition}{ \rm \cite{cover_elements_2012}}
The mutual information between the vectors of random variables $X^{n} \sim \Nc(\zerov,\Sigmam_{X\!X})$ and  $Y_{A}^{m} \sim \Nc(\zerov,\Sigmam_{Y_{A}\!Y_{A}})$ is given by
\begin{align}\label{MI}
 I(X^{n};Y^{m}_{A}) = \frac{1}{2} \log \frac{|\Sigmam_{X\!X}||\Sigmam_{Y_{A}\!Y_{A}}|}{|\Sigmam|},
\end{align}
where $\Sigmam$ is the covariance matrix of the joint distribution of $(X^{n}, Y_{A}^{m})$.
\end{proposition}

Substituting (\ref{KL}) and (\ref{MI}) in (\ref{Equ:WeightSum})  we can now pose the Gaussian attack construction as the following optimization problem:
\begin{align}\label{Equ:Weight_Mod}
\underset{\Sigmam_{A\!A} \in \Sc^{m}_{+}}{\text{min}} \ &-(\lambda - 1)\log|\Sigmam_{Y\!Y} + \Sigmam_{A\!A}| - \log |\Sigmam_{A\!A}+\sigma^{2}\Id| + \lambda \trace(\Sigmam_{Y\!Y}^{-\!1}\Sigmam_{A\!A}).  
\end{align}
We now proceed to solve the optimization problem in the preceding text. First, note that the optimization domain $\Sc^{m}_{+}$ is a convex set. The following proposition characterizes the convexity of the cost function.
\begin{proposition}\label{Stealth_Model}
Let $\lambda \geq 1 $.  Then the cost function in the optimization problem in (\ref{Equ:Weight_Mod}) is convex.
\end{proposition}
\begin{proof}
Note that the term $-\log|\Sigmam_{A\!A}+\sigma^{2}\Id|$ is a convex function on $\Sigmam_{A\!A} \in \Sc^{m}_{+} $ \cite{boyd_convex_2004}.
Additionally, $-(\lambda - 1)\log|\Sigmam_{Y\!Y} + \Sigmam_{A\!A}|$ is a convex function on $\Sigmam_{A\!A} \in \Sc^{m}_{+} $ when $\lambda \geq 1$. Since the trace operator is a linear operator and the sum of convex functions is convex, it follows that the cost function in (\ref{Equ:Weight_Mod}) is convex on $\Sigmam_{A\!A}\in \Sc^{m}_{+}$.
\end{proof}

\begin{theorem} \label{Stealth_OPT}
Let $\lambda \geq 1$. Then the solution to the optimization problem in (\ref{Equ:Weight_Mod})
 is
 \begin{equation}
 \label{eq:att_cons}
 \Sigmam_{A\!A}^{\star} = \frac{1}{\lambda}\Hm\Sigmam_{X\!X}\Hm^{ \sf T}.
 \end{equation}

\end{theorem}

\begin{proof}
Denote the cost function in (\ref{Equ:Weight_Mod})  by $f(\Sigmam_{A\!A})$. Taking the derivative of the cost function with respect to $\Sigmam_{A\!A}$ yields 
\begin{align}
 \frac{\partial  f(\Sigmam_{A\!A})}{\ \partial\Sigmam_{A\!A}} \hspace{-0.2em}
   = &\!-\!2(\lambda - 1)(\Sigmam_{Y\!Y} \!+\! \Sigmam_{A\!A})^{-\!1} \!-\! 2(\Sigmam_{A\!A} \! + \! \sigma^{2}\Id_M)^{-\!1} \!+ \! 2\lambda \Sigmam_{Y\!Y}^{-\!1} \!- \!\lambda \text{diag}(\Sigmam_{YY}^{-\!1}) \nonumber\\
& \ + (\lambda - 1)\text{diag}\left((\Sigmam_{Y\!Y} \!+\! \Sigmam_{A\!A})^{-\!1}\right) + \text{diag}\left((\Sigmam_{A\!A} \!+ \!\sigma^{2}\Id)^{-\!1}\right) . 
\end{align}
Note that the only critical point is $\Sigmam_{A\!A}^{\star} = \frac{1}{\lambda} \Hm\Sigmam_{X\!X}\Hm^{{\sf T}}$.
Theorem \ref{Stealth_OPT} follows immediately from combining this result with Proposition \ref{Stealth_Model}.
\end{proof}
\begin{corollary}\label{Cor_MI}
The mutual information between the vector of state variables and the vector of compromised measurements induced by the optimal attack construction is given by
\begin{align}
I(X^n;Y_{A}^m)=  \frac 1 2 \log\left | \Hm \mathbf{\Sigma}_{X\!X}\Hm^{{\sf T}} \left(\sigma^2\Id+\frac{1}{\lambda}\Hm \mathbf{\Sigma}_{X\!X}\Hm^{{{\sf T}}}\right)^{-1}\hspace{-0.7em}+\Id\right |. \label{Equ:mi}
\end{align}
\end{corollary}

Theorem \ref{Stealth_OPT} shows that the generalized stealth attacks share the same structure of the stealth attacks in \cite{Sun_information-theoretic_2017} up to a scaling factor determined by $\lambda$. The solution in Theorem \ref{Stealth_OPT} holds for the case in which $\lambda\geq 1$, and therefore, lacks full generality. However, the case in which $\lambda <1$ yields unreasonably high probability of detection \cite{Sun_information-theoretic_2017} which indicates that the proposed attack construction is indeed of practical interest in a wide range of state estimation settings.

%

The resulting attack construction is remarkably simple to implement provided that the information about the system is available to the attacker. Indeed, the attacker only requires access to the linearized Jacobian measurement matrix $\Hm$ and the second order statistics of the state variables, but the variance of the noise introduced by the sensors is not necessary. To obtain the Jacobian, a malicious attacker needs to know the topology of the grid, the admittances of the branches, and the operation point of the system. The second order statistics of the state variables on the other hand, can be estimated using historical data. In \cite{Sun_information-theoretic_2017} it is shown that the attack construction with a sample covariance matrix of the state variables obtained with historical data is asymptotically optimal when the size of the training data grows to infinity.

It is interesting to note that the mutual information in (\ref{Equ:mi}) increases monotonically with $\lambda$ and that it asymptotically converges to $I(X^n;Y^m)$, i.e. the case in which there is no attack. While the evaluation of the mutual information as shown in Corollary \ref{Cor_MI} is straightforward, the computation of the associated probability of detection yields involved expressions that do not provide much insight. For that reason, the probability of detection of optimal attacks is treated in the following section.

\subsection{Probability of Detection of Generalized Stealth Attacks}
\label{SEC:PDCI}

The asymptotic probability of detection of the generalized stealth attacks is governed by the KL divergence as described in (\ref{eq:Chernoff-Stein}). However in the non-asymptotic case, determining the probability of detection is difficult, and therefore, choosing a value of $\lambda$ that provides the desired probability of detection is a challenging task.
In this section we first provide a closed-form expression of the probability of detection by direct evaluation and show that the expression does not provide any practical insight over the choice of $\lambda$ that achieves the desired detection performance. That being the case, we then provide an upper bound on the probability of detection, which, in turn, provides a lower bound on the value of $\lambda$ that achieves the desired probability of detection.

\subsubsection{Direct Evaluation of the Probability of Detection}
\label{Sec:DEPD}
Detection based on the LRT with threshold $\tau$ yields a probability of detection given by
\begin{align}\label{Equ:PD_D}
{\sf P}_{\sf D} \eqdef \EE\left[\mathbbm{1}_{\left\{L(Y_A^m) \geq \tau\right\}}\right] .
\end{align}
 The following proposition particularizes the above expression to the optimal attack construction described in Section \ref{SEC:ITA}.

\begin{lemma}\label{Pro_PD}
The probability of detection of the LRT in (\ref{LHRT}) for the attack construction in (\ref{eq:att_cons}) is given by
\begin{align}\label{Pro_PD_1}
{\sf P}_{\sf D}(\lambda) & = \PP\left[{(U^p)}^{\sf T} \Deltam U^p \geq \lambda\left(2 \log\tau + \log \left|\Id + \lambda^{-1}\Deltam\right|\right)\right],
\end{align}
where $p=\text{rank} (\Hm \mathbf{\Sigma}_{X\!X}\Hm^{\sf T})$, $U^p\in\mathbb{R}^p$ is a vector of random variables with distribution $\Nc(\zerov,  \Id)$, and $\Deltam\in\mathbb{R}^{p\times p}$ is a diagonal matrix with entries given by $(\Deltam)_{i,i}=\lambda_i(\Hm \mathbf{\Sigma}_{X\!X}\Hm^{\sf T)}\lambda_i(\Sigmam_{Y\!Y}^{-\!1})$, where $\lambda_i(\Am)$ with $i=1,\ldots, p$ denotes the $i$-th eigenvalue of  matrix $\Am$ in descending order.
\end{lemma}

\begin{proof}
The probability of detection of the stealth attack is,
\begin{align}
\label{Equ:PD_1}
{\sf P_D}(\lambda) &= \ \int_{\Sc} \mathrm{d}P_{Y_{A}^{m}}\\
&= \frac{1}{(2\pi)^{\frac{m}{2}}\left|\Sigmam_{Y_{A}\!Y_{A}}\right|^{\frac12}}  \int_{\Sc}  \exp\left \{-\frac{1}{2}\yv^{\sf T}\Sigmam_{Y_{A}\!Y_{A}}^{-\!1}\yv\right \} \mathrm{d} \yv, \label{Equ:PD_2} 
\end{align}
where
\begin{equation}
\Sc =  \{ \yv \in \RR^{m} : L(\yv) \geq \tau\}.
\end{equation}
Algebraic manipulation yields the following equivalent description of the integration domain:
\begin{align}
\label{eq:int_dom_2}
\Sc   =  \left \{ \yv \in \RR^{m}: \yv^{\sf T} {\Deltam_0} \yv  \geq  2 \log\tau   +   \log |\Id + \Sigmam_{A\!A}\Sigmam_{Y\!Y}^{-\!1}|\right\},
\end{align}
with ${\Deltam_0} \eqdef \Sigmam_{Y\!Y}^{-\!1} - \Sigmam_{Y_{A}\!Y_{A}}^{-\!1} $. Let $\Sigmam_{Y\!Y}=\Um_{Y\!Y}\Lambdam_{Y\!Y}\Um_{Y\!Y}^{\sf T}$ where $\Lambdam_{Y\!Y}\in\RR^{m\times m}$ is a diagonal matrix containing the eigenvalues of $\Sigmam_{Y\!Y}$ in descending order and $\Um_{Y\!Y}\in\RR^{m\times m}$ is a unitary matrix whose columns are the eigenvectors of $\Sigmam_{Y\!Y}$ ordered matching the order of the eigenvalues. Applying the change of variable $\yv_{1} \eqdef \Um_{Y\!Y}\yv$ in (\ref{Equ:PD_2}) results in
\begin{align}
{\sf P_D}(\lambda) =
\frac{1}{(2\pi)^{\frac{m}{2}}\left|\Sigmam_{Y_{A}\!Y_{A}}\right|^{\frac12}}\int_{\Sc_1}\exp\left \{-\frac{1}{2}\yv^{\sf T}_1\Lambdam_{Y_{A}\!Y_{A}}^{-\!1}\yv_1\right \} \mathrm{d} \yv_1, \label{Equ:PD_3} 
\end{align}
{where $\Lambdam_{Y_{A}\!Y_{A}}\in\RR^{m\times m}$ denotes the diagonal matrix containing the eigenvalues of $\Sigmam_{Y_{A}\!Y_{A}}$ in descending order}.
Noticing that $\Sigmam_{Y\!Y}$, $\Sigmam_{A\!A}$ and $\Sigmam_{Y_{A}\!Y_{A}}$ are also diagonalized by $\Um_{Y\!Y}$,
the integration domain $\Sc_{1}$ is given by
\begin{align}
\Sc_{1}  =  \left\{  \yv_{1} \in  \RR^{m}  : \yv_{1}^{\sf T} \Deltam_{1} \yv_{1} \geq  2 \log\tau  +  \log |\Id  + \Lambdam_{A\!A}\Lambdam_{Y\!Y}^{-\!1}|  \right\},
\end{align}
where $\Deltam_{1} \eqdef \Lambdam_{Y\!Y}^{-\!1}-\Lambdam_{Y_{A}\!Y_{A}}^{-\!1}$ with $\Lambdam_{A\!A}$ denoting the diagonal matrix containing the eigenvalues of $\Sigmam_{{A}\!{A}}$ in descending order. Further applying the change of variable $\yv_{2} \eqdef \Lambdam_{Y_{A}\!Y_{A}}^{-\!\frac{1}{2}} \yv_{1}$ in (\ref{Equ:PD_3}) results in
\begin{equation}
{\sf P_D}(\lambda) =\frac{1}{\sqrt{(2\pi)^{m}}} \int_{\Sc_{2}} \exp\{-\frac{1}{2}\yv_{2}^{\sf T}\yv_{2}\} \mathrm{d} \yv_{2}, \label{Equ:PD_4}
\end{equation}
with the transformed integration domain given by
\begin{align}
\Sc_{2}  =  \left\{ \yv_{2} \in \RR^{m} : \yv_{2}^{\sf T} \Deltam_{2} \yv_{2} \geq 2 \log\tau  +   \log |\Id  + \Deltam_{2}|\right\},
\end{align}
with
\begin{align}
\Deltam_{2} \eqdef \Lambdam_{A\!A} \Lambdam_{Y\!Y}^{-\!1} .
\end{align}
Setting $\Deltam \eqdef \lambda\Deltam_2$ and noticing that $\text{rank}(\Deltam)=\text{rank} (\Hm \mathbf{\Sigma}_{X\!X}\Hm^{\sf T})$ concludes the proof.
\end{proof}

Notice that the left-hand term $(U^p)^{\sf T} \Deltam U^p$ in (\ref{Pro_PD_1}) is a weighted sum of independent $\chi^{2}$ distributed random variables with one degree of freedom where  the weights are determined by the diagonal entries of $\Deltam$ which depend on the second order statistics of the state variables, the Jacobian measurement matrix, and the variance of the noise; i.e. the attacker has no control over this term.
The right-hand side contains in addition $\lambda$ and $\tau$, and therefore,  the probability of attack detection is described as a function of the parameter $\lambda$. However, characterizing the distribution of the resulting random variable is not practical since there is no closed-form expression for the distribution of a positively weighted sum of independent $\chi^2$ random variables with one degree of freedom  \cite{bodenham_comparison_2016}.
Usually, some moment matching approximation approaches such as the Lindsay-Pilla-Basak method \cite{lindsay_moment-based_2000} are utilized to solve this problem but the resulting expressions are complex and the relation of the probability of detection with $\lambda$ is difficult to describe analytically following this course of action. In the following an upper bound on the probability of attack detection is derived. The upper bound is then used to provide a simple lower bound on the value $\lambda$ that achieves the desired probability of detection.

\subsubsection{Upper Bound on the Probability of Detection}

The following theorem provides a sufficient condition for $\lambda$ to achieve a desired probability of attack detection.

\begin{theorem}\label{pro_CI}
Let $\tau > 1$ be the decision threshold of the LRT. For any $t>0$ and $\lambda \geq \textnormal{max}\left(\lambda^{\star}(t),1\right)$ then the probability of attack detection satisfies
\begin{equation}
{\sf P_D}(\lambda)\leq e^{-t},
\end{equation}
where $\lambda^{*}(t)$ is the only positive solution of $\lambda$ satisfying
\begin{align}
 2 \lambda \log\tau  - \frac{1}{2\lambda}\textnormal{\trace}(\Deltam^2) -2\sqrt{\textnormal{\trace}({\Deltam}^{2})t} \!-\! 2||\Deltam||_{\infty} t = 0.
\end{align}
and $||\cdot||_{\infty}$ is the infinity norm.
\end{theorem}

\begin{proof}
We start with the result of Lemma \ref{Pro_PD} which gives
\begin{align}
\label{Equ:PD_CI1}
{\sf P_{D}}(\lambda)  = \PP  \left[{(U^p)}^{\sf T}  \Deltam U^p \geq  \lambda\left(2 \log\tau  +  \log \left|\Id   +  \lambda^{-1} \Deltam\right|\right)\right]. 
\end{align}
We now proceed to expand the term $\log \left|\Id + \lambda^{-1}\Deltam\right|$ using a Taylor series expansion resulting in
\begin{align}
\log \left|\Id + \lambda^{-1}\Deltam\right|  & = \sum_{i=1}^{p} \log\left(1+\lambda^{-1}(\Deltam)_{i,i}\right) \\
& = \sum_{i=1}^{p} \left(\sum_{j=1}^{\infty}  \left( \frac{\left(\lambda^{-1} (\Deltam)_{i,i}\right)^{2j-1}}{2j-1} - \frac{\left(\lambda^{-1}(\Deltam)_{i,i}\right)^{2j}}{2j}\right) \right).\label{Equ:Taylor_1} 
\end{align}
Because $(\Deltam)_{i,i} \leq 1,\textnormal{for} \; i=1, \ldots, p$, and $\lambda \geq 1$,
then
\begin{align}
\frac{\left(\lambda^{-1} (\Deltam)_{i,i}\right)^{2j-1}}{2j-1} - \frac{\left(\lambda^{-1}(\Deltam)_{i,i}\right)^{2j}}{2j} \geq 0, \ \textnormal{for}\; j \in \ZZ^+. 
\end{align}
Thus, (\ref{Equ:Taylor_1}) is lower bounded by the second order Taylor expansion, i.e.,
\begin{align}
\log  \left|\Id + \Deltam\right| &  \geq  \sum_{i=1}^{p} \left(\lambda^{-1}(\Deltam)_{i,i} -\frac{\left(\lambda^{-1} (\Deltam)_{i,i}\right)^2}{2}\right)\\
& = 
\label{eq:Taylor_2}
\frac{1}{\lambda}\trace(\Deltam) - \frac{1}{2\lambda^2}\trace(\Deltam^2). 
\end{align}
Substituting (\ref{eq:Taylor_2}) in (\ref{Equ:PD_CI1}) yields
\begin{align}
 {\sf P_{D}}(\lambda) \leq  \PP  \left[{(U^p)}^{\sf T}  \Deltam U^p  \geq  \trace(\Deltam)  +  2  \lambda\log\tau  -  \frac{1}{2\lambda}\trace(\Deltam^2)\right]  .  \label{Equ:PD_CI2} 
\end{align}
Note that $\EE\left[(U^p)^{\sf T} \Deltam U^p \right]=\trace(\Deltam)$, and therefore, evaluating the probability in (\ref{Equ:PD_CI2}) is equivalent to evaluating the probability of $(U^p)^{\sf T} \Deltam U^p$ deviating $2 \lambda  \log\tau - \frac{1}{2\lambda}\trace(\Deltam^2)$ from the mean. In view of this, the right-hand side in (\ref{Equ:PD_CI2}) is upper bounded by \cite{laurent_adaptive_2000,hsu_tail_2012}
\begin{align}
{\sf P_{D}}(\lambda) & \leq   \PP   \left[{(U^p)}^{\sf T}   \Deltam U^p   \geq   \trace(\Deltam)  + 2\sqrt{\trace({\Deltam}^{2})t}  + 2||\Deltam||_{\infty} t \right]  \leq e^{-t},\label{Equ:PD_CI5}
\end{align}
for $t >0$ satisfying
\begin{align} \label{Equ:CI_IE}
2 \lambda\log\tau  - \frac{1}{2\lambda}\trace(\Deltam^2) \geq 2\sqrt{\trace({\Deltam}^{2})t} + 2||\Deltam||_{\infty} t.
\end{align}
The expression in (\ref{Equ:CI_IE}) is satisfied with equality for two values of $\lambda$, one is strictly negative and the other one is strictly positive denoted by $\lambda^{\star}(t)$, when $\tau> 1$. The result follows by noticing that the left-hand term of (\ref{Equ:CI_IE}) increases monotonically for $\lambda>0$ and choosing $\lambda \geq \textnormal{max}\left(\lambda^{\star}(t),1\right)$. This concludes the proof.
\end{proof}

It is interesting to note that for large values of $\lambda$ the probability of detection decreases exponentially fast with $\lambda$. We will later show in the numerical results that the regime in which the exponentially fast decrease kicks in does not align with the saturation of the mutual information loss induced by the attack.

\subsection{Numerical Evaluation of Stealth Attacks} \label{Numerical_Simulation}

We evaluate the performance of stealth attacks in practical state estimation settings. n particular,
{the IEEE 14-Bus, 30-Bus and 118-Bus} test systems are considered in the simulation.
In state estimation with linearized dynamics, the Jacobian measurement matrix is determined by the operation point.
We assume a DC state estimation scenario \cite{abur_power_2004, grainger_power_1994}, and thus, we set {the resistances of the branches to $0$ and the bus voltage magnitude to $1.0$ per unit.} Note that in this setting it is sufficient to specify the network topology, the branch reactances, real power flow, and the power injection values to fully characterize the system. Specifically, we use the IEEE test system framework provided by MATPOWER \cite{zimmerman_matpower:_2011}.
{We choose the bus voltage angle to be the state variables, and use the power injection and the power flows in both directions as the measurements.}


%

As stated in Section \ref{Sec:DEPD}, there is no closed-form expression for the distribution of a positively weighted sum of independent $\chi^2$ random variables, which is required to calculate the probability of detection of the generalized stealth attacks as shown in Lemma \ref{Pro_PD}. For that reason, we use the Lindsay–Pilla–Basak method and the {MOMENTCHI2} package \cite{bodenham_momentchi2:_2016} to numerically evaluate the probability of attack detection.

The covariance matrix of the state variables is modelled as a Toeplitz matrix with exponential decay parameter $\rho$, where the exponential decay parameter $\rho$ determines the correlation strength between different entries of the state variable vector.
The performance of the generalized stealth attack is a function of weight given to the detection term in the attack construction cost function, i.e. $\lambda$, the correlation strength between state variables, i.e. $\rho$, and the Signal-to-Noise Ratio (SNR) of the power system which is defined as
\begin{equation}
\textnormal{SNR} \eqdef 10\log_{10}\left(\frac{\trace{(\Hm\Sigmam_{X\!X}\Hm^{\sf T}})}{m\sigma^2}\right).
\end{equation}

\begin{figure}[t!]
\centering
\includegraphics[scale=0.6]{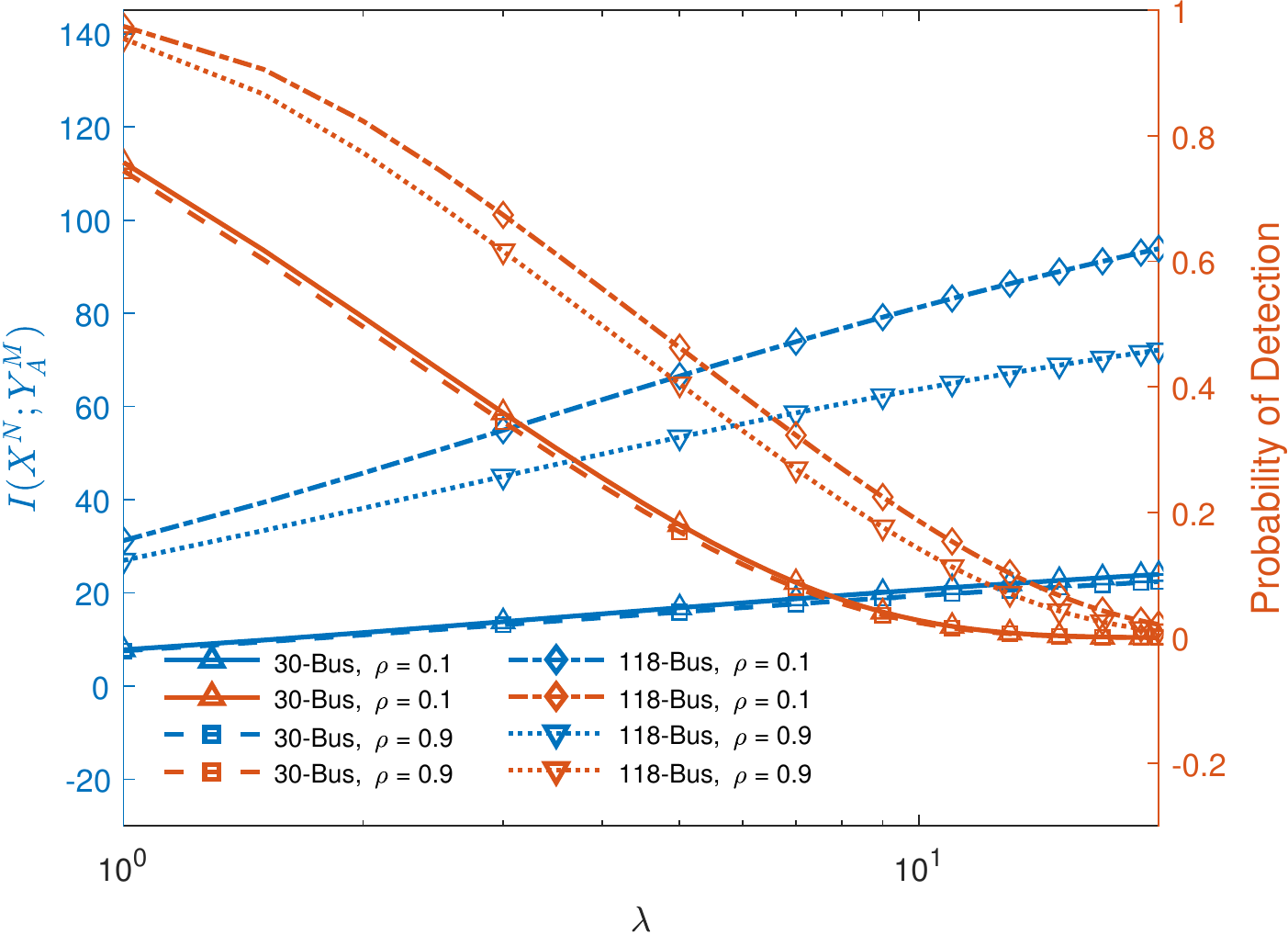}
\caption{ Performance of the generalized stealth attack in terms of mutual information  and probability of detection for different values of $\lambda$ and system size when $\rho = 0.1$, $\rho = 0.9$, $\textnormal{SNR} = 10 \ \textnormal{dB}$ and $\tau = 2$.}
\label{Fig:MIPD_lambda_log_S10}
\end{figure}

\begin{figure}[t!]
\centering
\includegraphics[scale=0.6]{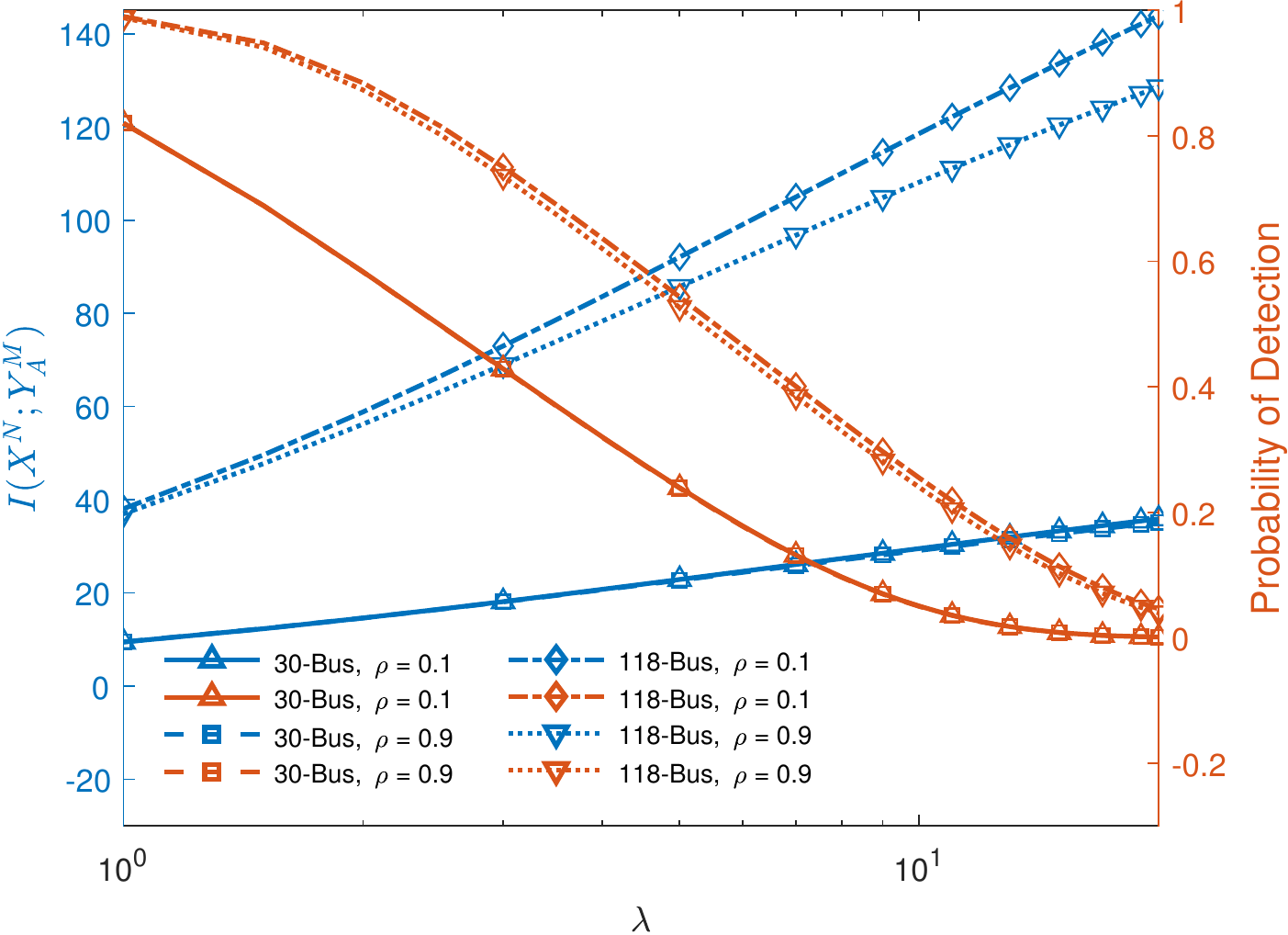}
\caption{ Performance of the generalized stealth attack in terms of mutual information  and probability of detection for different values of $\lambda$ and system size when $\rho = 0.1$, $\rho = 0.9$, $\textnormal{SNR} = 20 \ \textnormal{dB}$ and $\tau = 2$.}
\label{Fig:MIPD_lambda_log_S20}
\end{figure}


\begin{figure}[t!]
\centering
\includegraphics[scale=0.5]{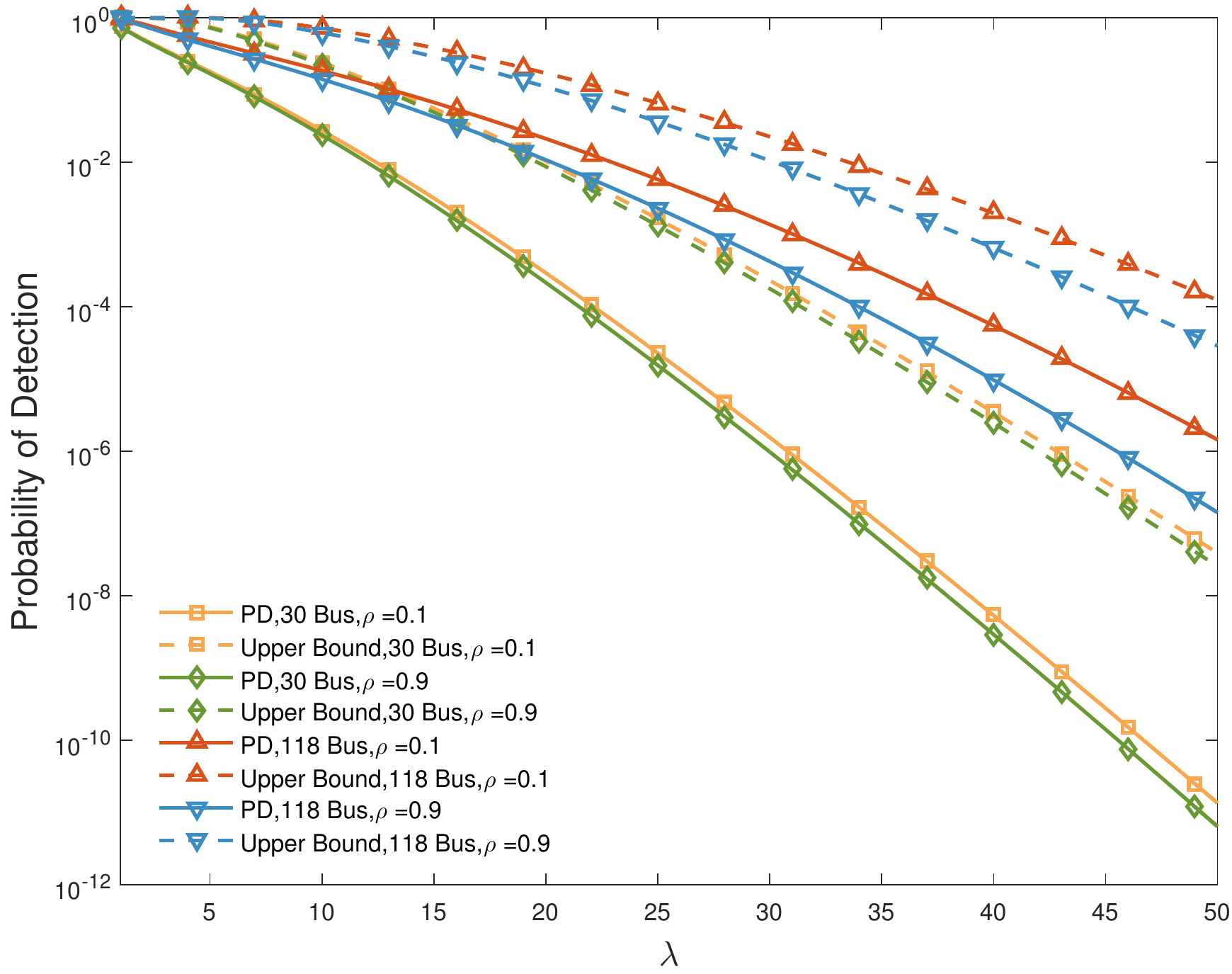}
\caption{Upper bound on probability of detection given in Theorem \ref{pro_CI} for different values of $\lambda$ when $\rho = 0.1 \ \text{or} \ 0.9$, $\textnormal{SNR} = 10 \ \textnormal{dB}$, and $\tau = 2$.}
\label{Fig:CI_Bound_S10}
\end{figure}

\begin{figure}[t!]
\centering
\includegraphics[scale=0.5]{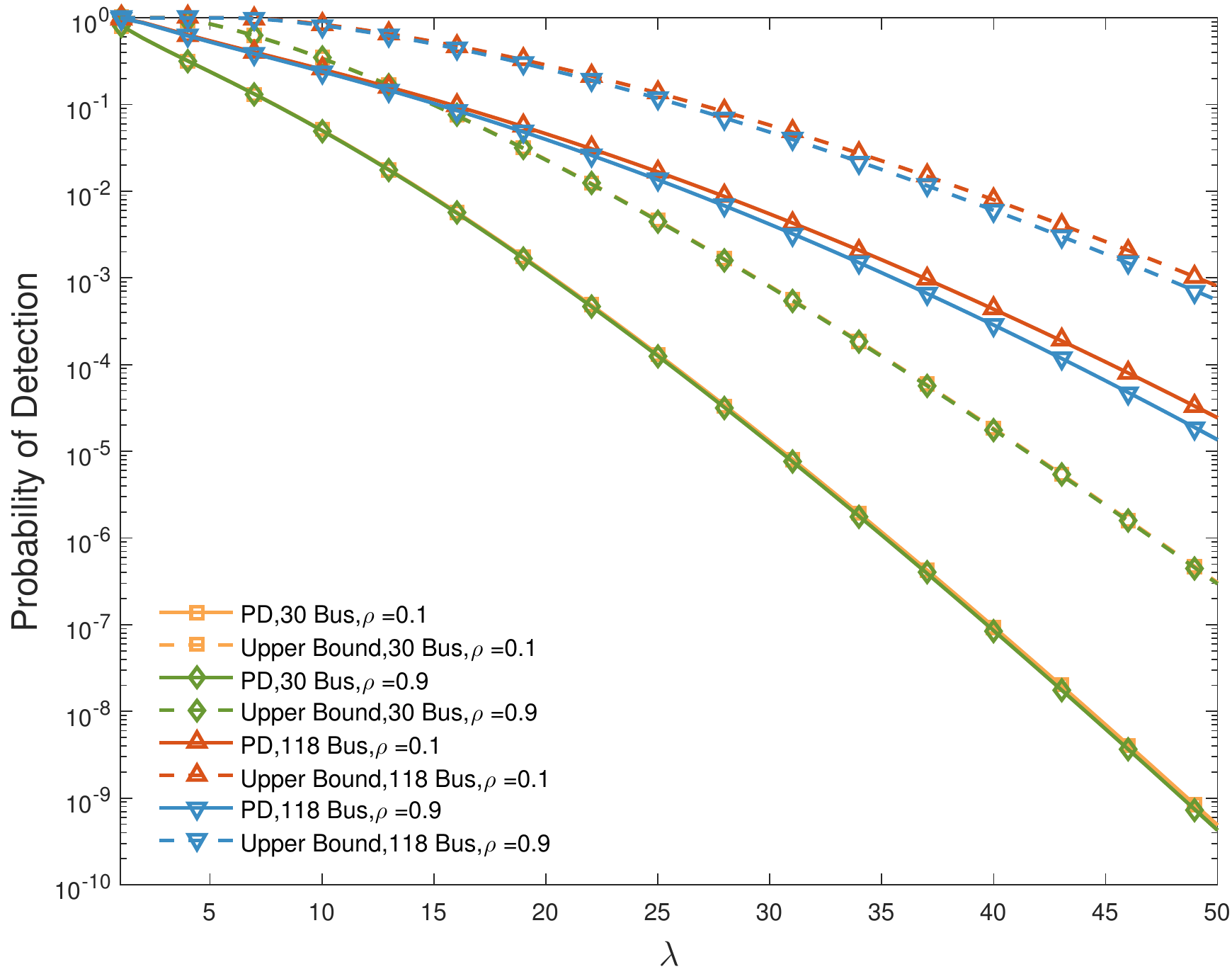}
\caption{Upper bound on probability of detection given in Theorem \ref{pro_CI} for different values of $\lambda$ when $\rho = 0.1 \ \text{or} \ 0.9$, $\textnormal{SNR} = 20 \ \textnormal{dB}$, and $\tau = 2$.}
\label{Fig:CI_Bound_S20}
\end{figure}

Fig. \ref{Fig:MIPD_lambda_log_S10} and Fig. \ref{Fig:MIPD_lambda_log_S20} depict the performance of the optimal attack construction for different values of $\lambda$ and $\rho$ with $\textnormal{SNR} = 10 \ \textnormal{dB}$ and $\textnormal{SNR} = 20 \ \textnormal{dB}$, respectively, when $\tau=2$. As expected, larger values of the parameter $\lambda$ yield smaller values of the probability of attack detection while increasing the mutual information between the state variables vector and the compromised measurement vector. We observe that the probability of detection decreases approximately linearly for moderate values of $\lambda$.
On the other hand, Theorem \ref{pro_CI} states that for large values of $\lambda$ the probability of detection decreases exponentially fast to zero.
However, for the range of values of $\lambda$ in which the decrease of probability of detection is approximately linear, there is no significant reduction on the rate of growth of mutual information. In view of this, the attacker needs to choose the value of $\lambda$ carefully as the convergence of the mutual information to the asymptote $I(X^n;Y^m)$ is slower than that of the probability of detection to zero.

The comparison between the 30-Bus and 118-Bus systems shows that for the smaller size system the probability of detection decreases faster to zero while the rate of growth of mutual information is smaller than that on the larger system.
This suggests that the choice of $\lambda$ is particularly critical in large size systems as smaller size systems exhibit a more robust attack performance for different values of $\lambda$.
The effect of the correlation between the state variables is significantly more noticeable for the 118-bus system.
While there is a performance gain for the 30-bus system in terms of both mutual information and probability of detection {due to the high correlation between the state variables}, the improvement is more noteworthy for the 118-bus case.
Remarkably, the difference in terms of mutual information between the case in which $\rho=0.1$ and $\rho=0.9$ increases as $\lambda$ increases which indicates that the cost in terms of mutual information of reducing the probability of detection is large in the small values of correlation.

The performance of the upper bound given by Theorem \ref{pro_CI} on the probability of detection for different values of $\lambda$ and $\rho$ when $\tau=2$ and $\textnormal{SNR} = 10 \ \textnormal{dB}$ is shown in Fig. \ref{Fig:CI_Bound_S10}. Similarly, Fig. \ref{Fig:CI_Bound_S20} depicts the upper bound with the same parameters but with $\textnormal{SNR} = 20 \ \textnormal{dB}$. As shown by Theorem \ref{pro_CI} the bound decreases exponentially fast for large values of $\lambda$. Still, there is a significant gap to the probability of attack detection evaluated numerically. This is partially due to the fact that our bound is based on the concentration inequality in \cite{laurent_adaptive_2000} which introduces a gap of more than an order of magnitude.  Interestingly, the gap decreases when the value of $\rho$ increases although the change is not  significant. More importantly, the bound is tighter for lower values of SNR for both 30-bus and 118-bus systems.

\section{Attack Construction with Estimated State Variable Statistics}

\subsection{Learning the Second-Order Statistics of the State Variables}

The stealth attack construction proposed in the preceding text requires perfect knowledge of the covariance matrix of the state variables and the linearized Jacobian measurement matrix. 
In \cite{sun_2019_learning}, the performance of the attack when the second-order statistics are not perfectly known by the attacker but the linearized Jacobian measurement matrix is known.
Therein, the partial knowledge is modelled by assuming that the attacker has access to a sample covariance matrix of the state variables.
Specifically, the training data consisting of $k$ state variable realizations $\{\xv^n_{i}\}^{k}_{i=1}$ is available to the attacker. That being the case the attacker computes the unbiased estimate of the covariance matrix of the state variables given by
\vspace{-0.5em}
\begin{equation}\label{Equ:SC}
\Sm_{X\!X} = \frac{1}{k-1} \sum_{i=1}^{k} \xv^n_{i} (\xv^n_{i})^{\sf T}.
\vspace{-0.5em}
\end{equation}
The stealth attack constructed using the sample covariance matrix follows a multivariate Gaussian distribution given by
\begin{equation}
\tilde{A}^{m} \sim \Nc (\zerov, \Sigmam_{\tilde{A}\!\tilde{A}}),
\end{equation}
where $\Sigmam_{\tilde{A}\!\tilde{A}} = \Hm\Sm_{X\!X}\Hm^{\sf T}$.

Because the sample covariance matrix in (\ref{Equ:SC}) is a random matrix with central Wishart distribution given by
\begin{align}\label{Equ:Wishart_S_XX}
 \vspace{-0.25em}
\Sm_{X\!X} \sim \frac{1}{k-1}W_n(k-1, \Sigmam_{X\!X}),
 \vspace{-0.25em}
\end{align}
the ergodic counterpart of the cost function in (\ref{Equ:Steallth_Obj}) is defined in terms of the conditional KL divergence given by
\begin{equation}
\label{eq:ergodic_cost}
\EE_{\Sm_{X\!X}}\!\left [D\left( P_{X^{n}Y_{A}^{m}|\Sm_{X\!X}}\|P_{X^{n}}P_{Y^{m}}\right)\right].
\end{equation}
The ergodic cost function characterizes the expected performance of the attack averaged over the realizations of training data. Note that the performance using the sample covariance matrix is suboptimal \cite{Sun_information-theoretic_2017} and that the ergodic performance converges asymptotically to the optimal attack construction when the size of the training data set increases.

\subsection{Ergodic Stealth Attack Performance}

In this section, we analytically characterize the ergodic attack performance defined in (\ref{eq:ergodic_cost}) by providing an upper bound using random matrix theory tools.
Before introducing the upper bound, some auxiliary results on the expected value of the extreme eigenvalues of Wishart random matrices are presented below.
 \vspace{-0.5em}
\subsubsection{Auxiliary Results in Random Matrix Theory}

\begin{lemma}\label{Pro:VarianceMaxSingular}
Let $\Zm_{l}$ be an $(k-1)\times l$ matrix whose entries are independent standard normal random variables, then
\begin{equation}
\textnormal{\var}\left(s_{max}(\Zm_{l})\right) \leq 1,
\end{equation}
where $\textnormal{\var}\left(\cdot\right)$ denotes the variance and $s_{max}(\Zm_{l})$ is the maximum singular value of $\Zm_{l}$.
\end{lemma}
\begin{proof}
Note that $s_{max}(\Zm_{l})$ is a 1-Lipschitz function of matrix $\Zm_{l}$, the maximum singular value of $\Zm_{l}$ is concentrated around the mean \cite[Proposition 5.34]{vershynin_introduction_2012} given by  $\EE[s_{max}(\Zm_{l})]$. Then for $t\geq 0$, it holds that
\begin{align}
\PP \! \left[ \left|s_{max}(\Zm_{l}) \!- \!\EE[s_{max}(\Zm_{l})]\right| > t \right] &\leq 2\exp\{-t^2/2\}\\
&\leq \exp\{1- t^2/2\}.
\end{align}
Therefore $s_{max}(\Zm_{l})$ is a sub-gaussian random variable with variance proxy $\sigma_{p}^2 \leq 1$.
The lemma follows from the fact that  $\textnormal{\var}\left(s_{max}(\Zm_{l})\right) \leq \sigma_{p}^2$.
\end{proof}

\begin{lemma}\label{Pro:MaxMinEigBound}
Let $\Wm_{l}$ denote a central Wishart matrix distributed as $\frac{1}{k-1} W_{l}(k-1,\Id)$, then the non-asymptotic expected value of the extreme eigenvalues of $\Wm_{l}$ is bounded by
\begin{align}\label{Equ:MinEigBound}
\left(1-\sqrt{l/(k-1)}\right)^2\leq \EE[\lambda_{min}(\Wm_{l})]
\end{align}
and
\vspace{-0.5em}
\begin{align}\label{Equ:MaxEigBound}
\EE[\lambda_{max}(\Wm_{l})] \leq \left(1+\sqrt{l/(k-1)}\right)^2 + 1/(k-1),
\end{align}
where $\lambda_{min}(\Wm_{l})$ and $\lambda_{max}(\Wm_{l})$ denote the minimum eigenvalue and maximum eigenvalue of $\Wm_{l}$, respectively.
\end{lemma}
\begin{proof}
Note that \cite[Theorem 5.32]{vershynin_introduction_2012}
\begin{equation}
\vspace{-0.1em}
\sqrt{k-1} -\sqrt{l}\leq \EE[s_{min}(\Zm_{l})] \label{Equ:Bound_MinV}
\vspace{-0.1em}
\end{equation}
and
\begin{equation}
\sqrt{k-1} + \sqrt{l} \geq \EE[s_{max}(\Zm_{l})] , \label{Equ:Bound_MaxV}
\end{equation}
where $s_{min}(\Zm_{l})$ is the minimum singular value of $\Zm_{l}$.
Given the fact that $\Wm_{l} = \frac{1}{k-1} \Zm_{l}^{\sf T}\Zm_{l}$, then it holds that
\begin{align}
\EE[\lambda_{min}(\Wm_{l})] &= \!\frac{\EE\!\left[{s_{min}(\Zm_{l})}^2\right]}{k-1}  \!\geq\! \frac{\EE\left[s_{min}(\Zm_{l})\right]^2}{k-1} \label{Equ:Exp_MinEig}
\end{align}
and
\vspace{-1em}
\begin{align}
\EE[\lambda_{max}(\Wm_{l})] \!= \!\frac{\EE\!\left[{s_{max}(\Zm_{l})}^2\right]}{k-1}  \!\leq \!\frac{\EE\left[s_{max}(\Zm_{l})\right]^2 \hspace{-0.2em} + 1}{k-1}, \label{Equ:Exp_MaxEig}
\end{align}
where (\ref{Equ:Exp_MaxEig}) follows from Lemma \ref{Pro:VarianceMaxSingular}.
Combining (\ref{Equ:Bound_MinV}) with (\ref{Equ:Exp_MinEig}), and (\ref{Equ:Bound_MaxV}) with (\ref{Equ:Exp_MaxEig}), respectively, yields the lemma.
\end{proof}

Recall the cost function describing the attack performance given in (\ref{eq:ergodic_cost}) can be written in terms of the covariance matrix $\Sigmam_{\tilde{A}\!\tilde{A}}$ in the multivariate Gaussian case with imperfect second-order statistics. The ergodic cost function that results from averaging the cost over the training data yields
\begin{align}
\label{eq:exp_stealth_cost_gauss}
\EE_{\Sm_{X\!X}}\!\left [D\left( P_{X^{n}Y_{A}^{m}|\Sm_{X\!X}}\|P_{X^{n}}P_{Y^{m}}\right)\right]& = \! \frac{1}{2}\EE\!\left[\trace(\Sigmam_{Y\!Y}^{-\!1}\Sigmam_{\tilde{A}\!\tilde{A}})\!-\!\log |\Sigmam_{\tilde{A}\!\tilde{A}}\!+\!\sigma^{2}\Id|\!-\!\log |\Sigmam_{Y\!Y}^{-\!1}|\right] \\
& = \! \frac{1}{2}\!\Big(\!\trace \! \left(\Sigmam_{Y\!Y}^{-\!1}\Sigmam^\star_{A\!A}\right)\!-\!\log \! \left|\!\Sigmam_{Y\!Y}^{-\!1}\!\right|\!-\! \EE\!\left[\log \! |\Sigmam_{\tilde{A}\!\tilde{A}}\!+\!\sigma^{2}\Id|\right]\!\Big).
\vspace{-0.25em}
\end{align}
The assessment of the ergodic attack performance boils down to evaluating the last term in (\ref{eq:exp_stealth_cost_gauss}). Closed form expressions for this term are provided in \cite{alfano_capacity_2004} for the same case considered in this paper. However, the resulting expressions are involved and are only computable for small dimensional settings. For systems with a large number of dimensions the expressions are computationally prohibitive. To circumvent this challenge we propose a lower bound on the term that yields an upper bound on the ergodic attack performance. Before presenting the main result we provide the following auxiliary convex optimization result.


\begin{lemma} \label{Lemma:logdet_Inv_Wishart}
{Let $\Wm_{p}$ denote a central Wishart matrix distributed as $\frac{1}{k-1} W_{p}(k-1,\Id)$} and let $\Bm = \textnormal{diag} (b_{1}, \dots, b_{p})$ denote a positive definite diagonal matrix.
Then
\begin{align}
\vspace{-2em}
\EE\left[\log \left|\Bm + \Wm_{p}^{-1} \right|\right] \geq \sum_{i=1}^{p} \log\left(b_{i} + 1/x_{i}^{\star}\right),
\vspace{-4em}
\end{align}
where $x_{i}^{\star}$ is the solution to the convex optimization problem given by
\vspace{-1em}
\begin{align}
\underset{\left \lbrace x_{i}\right\rbrace_{i=1}^p}{\textnormal{min}} \ & \sum_{i=1}^{p} \log\left( b_{i} + 1/x_{i} \right)\label{Equ:BUPN_1}\\
s.t. \ \ & \sum_{i=1}^{p} x_{i}  = p \label{Equ:BUPN_4} \\
& \textnormal{max}\left(x_{i}\right) \leq  \left(1+\sqrt{p/(k-1)}\right)^2 + 1/(k-1)  \label{Equ:BUPN_5} \\
& \textnormal{min}\left( x_{i}\right) \geq  \left(1-\sqrt{p/(k-1)}\right)^2 .\label{Equ:BUPN_6}
\vspace{-0.5em}
\end{align}
\end{lemma}

\begin{proof}
Note that
\begin{align}
\EE\left[\log \left|\Bm + \Wm_{p}^{-1} \right|\right] &= \sum_{i=1}^{p} \EE \left[\log \left( b_{i} + \frac{1}{\lambda_{i}(\Wm_{p})}\right)\right] \label{Equ:Convex_Opt_1}\\
&\geq \sum_{i=1}^{p} \log \left( b_{i} +  \frac{1}{\EE [\lambda_{i}(\Wm_{p})]}\right), \label{Equ:Convex_Opt_2}
\end{align}
where in (\ref{Equ:Convex_Opt_1}), $\lambda_{i}(\Wm_{p})$ is the $i$-th eigenvalue of $\Wm_{p}$ in decreasing order;
(\ref{Equ:Convex_Opt_2}) follows from Jensen's inequality due to the convexity of $\log\left(b_{i} +\frac{1}{x}\right)$ for $x > 0$.
Constraint (\ref{Equ:BUPN_4}) follows from the fact that $\EE [\textnormal{trace} (\Wm_{p})] = p$, and constraints  (\ref{Equ:BUPN_5}) and  (\ref{Equ:BUPN_6}) follow from Lemma \ref{Pro:MaxMinEigBound}. This completes the proof.
\end{proof}

\subsubsection{Upper Bound on the Ergodic Stealth Attack Performance}
The following theorem provides a lower bound for the last term in  (\ref{eq:exp_stealth_cost_gauss}), and therefore, it enables us to
upper bound the ergodic stealth attack performance.

\begin{theorem}\label{Theorem:NonAsym_1}
{Let $\Sigmam_{\tilde{A}\!\tilde{A}}=\Hm\Sm_{X\!X}\Hm^{\sf T}$ with $\Sm_{X\!X}$ distributed as $\frac{1}{k-1}W_n(k-1, \Sigmam_{X\!X})$} and denote by $\Lambdam_{p} = \textnormal{diag} (\lambda_{1}, \dots, \lambda_{p})$ the diagonal matrix containing the nonzero eigenvalues in decreasing order.
Then
\begin{align}
\EE\!\left[\log \! |\Sigmam_{\tilde{A}\!\tilde{A}}\!+\!\sigma^{2}\Id|\right] \geq  \left(\sum_{i=0}^{p-1} \psi (k-1-i) \right)- p\log(k-1)+ \sum_{i=1}^{p} \log\left(\frac{\lambda_{i}}{\sigma^{2}} + \frac{1}{\lambda_{i}^{\star}}\right)+2m\log\sigma,
\end{align}
where $\psi (\cdot)$ is the Euler digamma function, $p=\textnormal{rank}(\Hm\Sigmam_{X\!X}\Hm^{\sf T})$, and
$\lbrace\lambda_{i}^{\star}\rbrace_{i=1}^p$ is the solution to the optimization problem given by (\ref{Equ:BUPN_1}) - (\ref{Equ:BUPN_6}) with $b_{i} = \frac{\lambda_{i}}{\sigma^{2}}, \textnormal{for}\; i = 1, \dots,p$.
\end{theorem}

\begin{proof}
We proceed by noticing that
\begin{align}
\EE\!\left[\log \! |\Sigmam_{\tilde{A}\!\tilde{A}}\!+\!\sigma^{2}\Id|\right]
& \ =\EE\left[\log \left|\frac{1}{(k-1)\sigma^{2}}  \Zm_{m}^{\sf T} \Lambdam\Zm_{m} +\Id\right|\right] + 2m \log \sigma\label{Equ:LBN_1}\\
& \ = \EE\left[\log \left|\frac{\Lambdam_{p}}{\sigma^{2}}  \frac{\Zm_{p}^{\Tt} \Zm_{p}}{k-1} +\Id\right|\right] + 2m \log \sigma\label{Equ:LBN_2}\\
& \ = \! \EE  \! \left[\!\log  \! \left|\frac{\Zm_{p}^{\Tt} \Zm_{p}}{k-1}\right|  \! +  \! \log \!\left|\frac{\Lambdam_{p}}{\sigma^{2}} \! \! + \! \left( \frac{\Zm_{p}^{\Tt} \Zm_{p}}{k-1}\right)^{-\!1}  \! \right|\right] \!+  \!2m\log \sigma \label{Equ:LBN_3}
\\
&\ \geq  \left(\sum_{i=0}^{p-1} \psi (k-1-i) \right)- p\log(k-1)
 + \sum_{i=1}^{p} \log\left(\frac{\lambda_{i}}{\sigma^{2}} + \frac{1}{\lambda_{i}^{\star}}\right)+2m\log\sigma, \label{Equ:LBN_5}
\end{align}
where in (\ref{Equ:LBN_1}), $\Lambdam$ is a diagonal matrix containing the eigenvalues of $\Hm\Sigmam_{X\!X}\Hm^{\Tt}$ in decreasing order;
(\ref{Equ:LBN_2}) follows from the fact that $p=\textnormal{rank}(\Hm\Sigmam_{X\!X}\Hm^{\Tt})$;
{(\ref{Equ:LBN_5}) follows from \cite[Theorem 2.11]{tulino_random_2004} and Lemma \ref{Lemma:logdet_Inv_Wishart}.}
This completes the proof.
\end{proof}

\begin{theorem}\label{Theorem:NonAsym_2}
The ergodic attack performance given in (\ref{eq:exp_stealth_cost_gauss}) is upper bounded by
\begin{align}
\EE\left[f(\Sigmam_{\tilde{A}\!\tilde{A}} )\right] &\leq  \frac{1}{2}\!\Bigg(\!\textnormal{trace} \! \left(\Sigmam_{Y\!Y}^{-\!1}\Sigmam^\star_{A\!A}\right)\!-\!\log \! \left|\!\Sigmam_{Y\!Y}^{-\!1}\!\right| - 2m \log \sigma \Bigg. 
\\& \quad
 - \bigg(\sum_{i=0}^{p-1} \psi (k-1-i)\! \bigg) \!+\! p\log(k-1)  \\
 & \quad \Bigg. -\sum_{i=1}^{p} \log\left(\frac{\lambda_{i}}{\sigma^{2}}+\frac{1}{\lambda_{i}^{\star}}\right)\Bigg).
\end{align}
\end{theorem}
\begin{proof}
The proof follows immediately from combing Theorem \ref{Theorem:NonAsym_1} with (\ref{eq:exp_stealth_cost_gauss}).
\end{proof}


Fig.\ref{Fig:UB_Rho_30_SNR20_NonAsy} depicts the upper bound in Theorem \ref{Theorem:NonAsym_2} as a function of number of samples for $\rho =0.1$ and $\rho=0.8$ when $\textnormal{SNR} = 20 \ \textnormal{dB}$.
Interestingly, the upper bound in Theorem \ref{Theorem:NonAsym_2} is tight for large values of the training data set size for all values of the exponential decay parameter determining the correlation.

\begin{figure}[t!]
  \centering
  \includegraphics[scale =0.5]{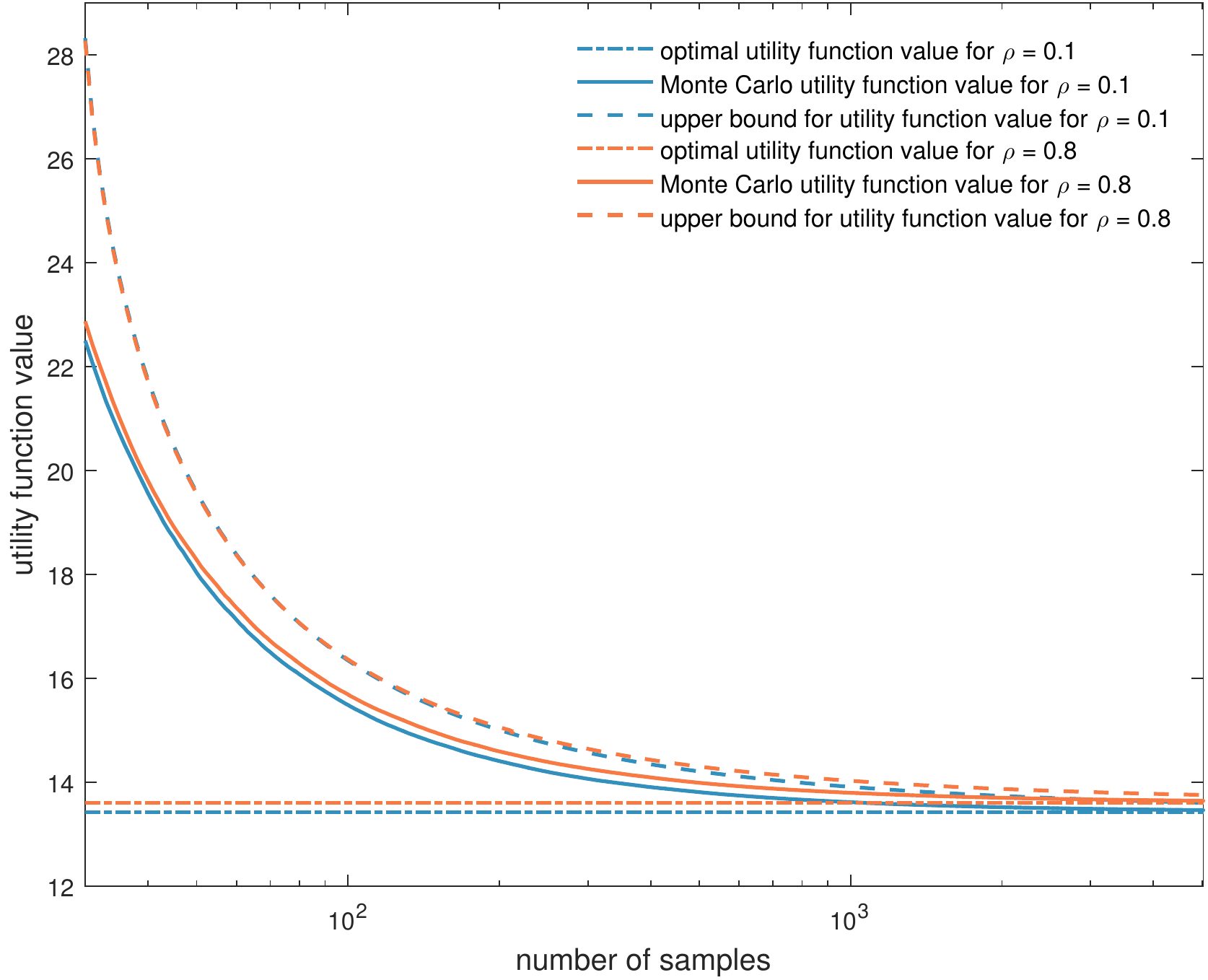}\\
  \caption{Performance of the upper bound in Theorem \ref{Theorem:NonAsym_2} as a function of number of sample for $\rho =0.1$ and $\rho=0.8$ when $\textnormal{SNR} = 20 \ \textnormal{dB}$.}\label{Fig:UB_Rho_30_SNR20_NonAsy}
\end{figure}
\section{Conclusions} \label{Conclusion}

We have cast the state estimation problem in a Bayesian setting and shown that the attacker can construct data-injection attacks that exploit prior knowledge about the state variables. In particular, we have focused in multivariate Gaussian random processes to describe the state variables and proposed two attack construction strategies: determinis- tic attacks and random attacks.

The deterministic attack is specified by the power system and the statistical structure of the state variables. The attack problem is cast as a multiobjective optimization prob- lem in which the attacker aims to simultaneously minimize the MSE distortion induced by the injection vector and the probability of the attack being detected using a LRT. Within this setting, we have characterized the tradeoff between the achievable distortion and probability of detection by deriving optimal centralized attack constructions for a given distortion and probability of detection pair. We have then extended the investi- gation to decentralized scenarios in which several attackers construct their respective attack without coordination. In this setting, we have posed the interaction between the attackers in a game-theoretic setting. We show that the proposed utility function results in a setting that can be described as a potential game that allows us to claim the existence of an NE and the convergence of BRD to an NE.

The random attack produces different attack vectors for each set of measurements that are reported to the state estimator. The attack vectors are generated by sampling a defined attack vector distribution that yields attack vector realizations to be added to the measurements. The attack aims to disrupt the state estimation process by minimizing the mutual information between the state variables and the altered measurements while minimizing the probability of detection. The rationale for posing the attack construction in information-theoretic terms stems from the fundamental character that information measures grant to the attack vector. By minimizing the mutual information, the attacker limits the performance of a wide range of estimation, detection, and learning options for the operator. We conclude the chapter by analyzing the impact of imperfect second- order statistics about the state variables in the attack performance. In particular, we consider the case in which the attacker has access to a limited set of training state variable observations that are used to produce the sample covariance matrix of the state variables. Using random matrix theory tools we provide an upper bound on the ergodic attack performance.

This work was supported in part by the European Commission under Marie Skodowska–Curie Individual Fellowship No. 659316 and in part by the Agence Nationale de la Recherche (ANR, France) under Grant ANR-15-NMED-0009-03 and the China Scholarship Council (CSC, China).